\documentclass[preprint,12pt]{elsarticle}

\usepackage{regexpatch}
\makeatletter
\regexpatchcmd{\ps@pprintTitle}
  {\cE\}\ \cE\}}{\cE\}\cE\}} 
  {}{\typeout{Failed to patch \ps@pprintTitle}}
\makeatother

\usepackage{graphicx}
\usepackage{tikz}
\usepackage{amssymb}

\usepackage{lineno}
\usepackage[bookmarks,bookmarksnumbered]{hyperref}
\hypersetup{colorlinks = true,linkcolor = blue,anchorcolor =red,citecolor = blue,filecolor = red,urlcolor = red}
\hypersetup{pdfauthor={Name}}

\usepackage{xcolor}
\usepackage{graphicx}
\usepackage[version=4]{mhchem}
\usepackage{amssymb}
\usepackage{amsmath}
\usepackage{amsfonts}
\def\eps{\ensuremath\varepsilon}
\def\cE{\mathbb{E}}
\usepackage{lineno}
\usepackage[margin=1.0in]{geometry}
\usepackage{natbib}
\usepackage{enumitem}

\usepackage{nicefrac}
\usepackage{appendix}

\usepackage{subfigure}

\usepackage{amsthm}
\theoremstyle{definition} 
\newtheorem{theorem}{Theorem}
\newtheorem{proposition}{Proposition}

\newtheorem{definition}{Definition}

\journal{arXiv}

\begin{document}

\begin{frontmatter}

\title{Accuracy versus Predominance: Reassessing the validity of the quasi-steady-state approximation\footnote{In 
memory of Claudio Mendoza (1951---2024), a scientist, an intellectual, mentor and friend, whose insights and 
collaboration enriched our understanding of enzyme kinetics and the polymerase chain reaction.}}



\author[label1]{Kashvi Srivastava}
\author[label2]{Justin Eilertsen}
\author[label1]{Victoria Booth}
\author[label3]{Santiago Schnell}
\ead{santiago.schnell@nd.edu}

\address[label1]{Department of Mathematics, University of Michigan, Ann Arbor, MI 48109, USA}
\address[label2]{Mathematical Reviews, American Mathematical Society, 
416 $4th$ Street, Ann Arbor, MI 48103, USA}
\address[label3]{Department of Biological Sciences and Department of Applied \& 
Computational Mathematics \& Statistics, University of Notre Dame, Notre Dame, IN 46556, USA}

\begin{abstract}
The application of the standard quasi-steady-state approximation to the Michaelis–Menten
reaction mechanism is a textbook example of biochemical model reduction, derived using
singular perturbation theory. However, determining the specific biochemical conditions that
dictate the validity of the standard quasi-steady-state approximation remains a challenging
endeavor. Emerging research suggests that the accuracy of the standard quasi-steady-state
approximation improves as the ratio of the initial enzyme concentration, $e_0$, to the Michaelis
constant, $K_M$ , decreases. In this work, we examine this ratio and its implications for the
accuracy and validity of the standard quasi-steady-state approximation as 
compared to other quasi-steady-state reductions in its proximity. Using standard
tools from the analysis of ordinary differential equations, we show that while $e_0/K_M$ provides
an indication of the standard quasi-steady-state approximation’s asymptotic accuracy, the
standard quasi-steady-state approximation’s predominance relies on a small ratio of $e_0$ to
the Van Slyke-Cullen constant, $K$. Here, we define the predominance of a 
quasi-steady-state reduction when it offers the highest approximation accuracy among other 
well-known reductions with overlapping validity conditions. We conclude that the magnitude 
of $e_0/K$ offers the most accurate measure of the validity of the standard quasi-steady-state 
approximation.
\end{abstract}


\end{frontmatter}

\section{Introduction} \label{S:intro}
The Michaelis--Menten (MM) reaction mechanism, a cornerstone of enzyme kinetics, describes the 
irreversible catalysis of a substrate, $S$, into product, $P$, by an enzyme, $E$:
\begin{align}\label{MM}
\ce{S + E <=>[$k_1$][$k_{-1}$] C ->[$k_2$] E + P}
\end{align}
where $C$ represents an enzyme-substrate intermediate complex \cite{michaelis1913kinetik}. The 
simplest form of the mass action equations modeling the dynamics of the substrate, complex, enzyme, 
and product concentrations ($s$, $c$, $e$, and $p$, respectively) is a two-dimensional system of 
ordinary differential equations (ODEs):
\begin{subequations}\label{mmMA}
\begin{align}
\dot{s} &= -k_1(e_0-c)s+k_{-1}c,\label{mms}\\
\dot{c} &=  \;\;\; k_1(e_0-c)s-(k_{-1}+k_2)c,\label{mmc} 
\end{align}
\end{subequations}
with $\dot{e}=-\dot{c}$, $\dot{p}=k_2c$, where ``$\dot{\phantom{x}}$" denotes differentiation with 
respect to time, $t$. The constant, $e_0$, is the total enzyme concentration, a conserved quantity. 

Closed-form solutions to \eqref{mmMA} are unattainable due to quadratic nonlinearities. Consequently, 
reduced equations approximating the long-time dynamics of \eqref{mmMA} are often employed to both 
elucidate the kinetics of \eqref{MM} and serve as a forward model for parameter estimation from in 
vitro progress curve experiments via nonlinear regression \cite{schnell2003century, Stroberg2016, choi2017beyond}. 
The standard quasi-steady-state approximation (sQSSA) to the mass action equations associated with 
the MM reaction mechanism is a well-known and extensively studied reduced model in 
biochemical kinetics \cite{schnell2003century,schnell2014validity}. Simply put, the sQSSA is 
a reduced system that approximates the depletion of substrate:
\begin{equation}\label{sQSSA}
\dot{s} = -\cfrac{k_2e_0s}{K_M+s}, \quad s(0)=s_0,
\end{equation}
where $K_M = (k_{-1}+k_2)/k_1$ is the Michaelis constant and $k_2\,e_0$ is the limiting rate. The
Michaelis constant is composite of the equilibrium constant $K_S=k_{-1}/k_1$ and the Van Slyke-Cullen 
constant $K=k_2/k_1$, and can be expressed as $K_M = K_S + K$. We assume initial conditions for 
\eqref{mmMA} lie on the s-axis, i.e., $(s,c)(0)=(s_0,0)$. Under this assumption, we can equip 
\eqref{sQSSA} with the initial condition $s(0)=s_0$. \footnote{Loosely speaking, 
solutions to the mass action system converge to solutions of the reduced system as 
$\varepsilon \to 0$ as long as initial conditions are sufficiently close to the critical manifold. 
Please visit~\ref{sec:appendix_initialdata} for details on the initial conditions of quasi-steady-state 
reductions.}

Geometric singular perturbation theory provides mathematical justification for \eqref{sQSSA}: with 
$(s,c)(0)=(s_0,0)$, the solution to the first component of the mass action 
system~\eqref{mms} converges to the approximation~\eqref{sQSSA} 
in the limit as $e_0 \to 0$, provided all other parameters, $k_{-1}, k_2$ and $k_1$, are bounded above 
and below by positive constants \cite[see, e.g.][]{Goeke2015,Eilertsen:2024:Unreasonable}. For further 
information on geometric singular perturbation theory,
see \citet{kaper2002asymptotic, hek2010geometric, jones1995geometric} and \citet{Wechselberger2020}. 

While the mathematical rationale behind \eqref{sQSSA} is well-established, applying \eqref{sQSSA} 
presents a challenge: determining how small $e_0$ must be to ensure \eqref{sQSSA} is sufficiently 
accurate. This assessment is complicated by the relative nature of ``small.'' $e_0$ is only small 
compared to another quantity, but what is that quantity?  Emerging research 
 \cite[see, e.g.][]{eilertsen2023natural,eilertsen2024rigorous,Eilertsen:2024:Unreasonable} 
suggests that the sQSSA is accurate whenever $e_0 \ll K_M$, or $\eps_{RS} = e_0/K_M \ll 1$, which we 
refer to as the Reich-Sel'kov condition \cite{reich1974mathematical}. Thus, the ``other'' quantity 
is the Michaelis constant, $K_M$, and the sQSSA is considered \textit{valid} -- in the sense of 
accuracy -- when the Reich-Sel'kov condition holds.

Beyond accuracy, it is crucial to determine conditions that ensure the \textit{predominance} of the 
sQSSA~\eqref{sQSSA}. When the reaction unfolds away from the singular limit (i.e., with $0<e_0$), 
multiple reduced models may be accurate.  However, it is not obvious that the Reich-Sel'kov condition 
guarantees that the sQSSA is the predominant reduction among the nearby Fenichel 
reductions. Other reduced models may be more 
accurate\footnote{In this context, \textit{accuracy} refers to the ability of the reduction to 
approximate the flow on the slow manifold. We are choosing to ignore errors that arise in the 
approach to the slow manifold.} 
than the sQSSA even when $e_0\ll K_M$. Therefore, identifying the most accurate 
approximation among the available reduced models with overlapping validity conditions 
is essential, and this most accurate model should be 
labeled as the predominant reduction (see, Section~\ref{SEC2} for a better 
understanding of the class of reduced models considered).

This paper aims to deepen our understanding of the Reich-Sel'kov condition and distinguish between 
accuracy and predominance concerning the sQSSA~\eqref{sQSSA}.  To our knowledge, this distinction has 
not been previously addressed in the literature.  Using rigorous phase-plane analysis and methods from 
ODE analysis, we demonstrate that while the Reich-Sel'kov condition provides a general indication of 
the sQSSA's accuracy, it does not ensure its predominance, which requires a more restrictive criterion.  
We derive this qualifier and challenge the notion that the validity of the sQSSA solely pertains to its 
accuracy.  We argue that validity should encompass both the accuracy and predominance of the sQSSA.  
In essence, the qualifier establishing the sQSSA's validity is more restrictive than any previously 
reported condition.

\section{Fenichel Theory, TFPV Theory, and Literature Review}\label{SEC2}
This section reviews Fenichel theory and Tikhonov-Fenichel Parameter Value (TFPV) theory, placing the 
problem within a precise mathematical context. We also examine recent literature on the validity of 
the sQSSA~\eqref{sQSSA}.

\subsection{A brief introduction of Fenichel theory}
Geometric singular perturbation theory, also known as Fenichel theory \cite{Fenichel1971, Fenichel1979}, 
provides the rigorous foundation for deriving the sQSSA. Consider a perturbed dynamical system:
\begin{equation}\label{singular}
\dot{x} = f(x,\varepsilon),\quad f:\mathbb{R}^n\times \mathbb{R} \to \mathbb{R}^n
\end{equation}
where $\varepsilon \in \mathbb{R}$ is a small parameter close to zero.\footnote{We assume $f(x,\varepsilon)$
is sufficiently smooth in both arguments.} If the set of singularities:
\begin{equation}
S_0 := \{x\in \mathbb{R}^n: f(x,0)=0\}
\end{equation}
is a non-empty $k$-dimensional submanifold of $\mathbb{R}^n$ where $0< k < n$, then \eqref{singular} is 
\textit{singularly perturbed}~\cite{Wechselberger2020}. The rank of the Jacobian with respect to 
$x$, $D_1f(x,0)$, is constant and equal to $n-k$ along $S_0$.

\subsubsection{Persistence and Flow}
A key result from Fenichel theory states that if $S_0$ is compact and normally hyperbolic -- meaning the 
$n-k$ non-zero eigenvalues of $D_1f(x,0)$  are strictly bounded away from the imaginary axis along $S_0$ --
then $S_0$ persists. This means an invariant manifold, $S_0^{\varepsilon}$, exists for 
$\varepsilon \in [0,\varepsilon_0)$. Expressing $S_0^{\varepsilon}$ as a graph over the 
first $k$ coordinates of $x$
\begin{equation}
S_0^{\varepsilon}:=\{(z,y,\varepsilon)\in \mathbb{R}^k \times \mathbb{R}^{n-k}\times \mathbb{R}: y = h(z,\varepsilon)\}, 
\end{equation}
where $x:=(z,y)$ and $z\in\mathbb{R}^k,\, y\in \mathbb{R}^{n-k}$, the flow on $S_0^{\varepsilon}$ is satisfies:
\begin{equation}\label{slowflow}
\dot{z}_i = f_i(z,h(z,\varepsilon),\varepsilon),\quad 1\leq i \leq k
\end{equation}
where the $f_i$'s are the component functions of $f$. Equation~\eqref{slowflow} describes the evolution 
of $k$ independent variables. The evolution of the remaining $n-k$ variables depends 
on the evolution of $z$:
\begin{equation}
\dot{y} = D_1h(z,\varepsilon)\dot{z},\quad D_1h(z,\varepsilon) \in \mathbb{R}^{(n-k)\times k}.
\end{equation}

\subsubsection{Critical Manifolds and Stability}
In singular perturbation theory, the submanifold $S_0$ of stationary points is the \textit{critical 
manifold}. The persistence of a normally hyperbolic $S_0$ asserts the existence of a normally hyperbolic 
invariant manifold, $S_0^{\varepsilon}$, and also the persistence of its stable and unstable manifolds,
$W^s(S_0)$ and $W^u(S_0)$. Thus, $S_0^{\varepsilon}$ inherits the stability properties as 
$S_0$.\footnote{Depending on $S_0$, $S_0^{\varepsilon}$ will either be attractive 
$W^u(S_0)=\emptyset \implies W^u(S_0^\varepsilon)=\emptyset$, repulsive 
$W^s(S_0)=\emptyset \implies W^s(S_0^\varepsilon)=\emptyset$, 
or will be of saddle type and come equipped with both stable and unstable manifolds} If $S_0$ is normally 
hyperbolic and attracting, \eqref{slowflow} describes the long-time evolution of \eqref{singular} because
trajectories near $S_0^{\varepsilon}$ contract towards it exponentially.

\subsubsection{Approximating the Slow Flow}
Fenichel's initial results require explicit knowledge of $S_0^{\varepsilon}$, which is often unavailable 
in practice. To address this, \citet{Fenichel1979} proposed a method to approximate the flow on 
$S_0^{\varepsilon}$ without explicitly knowing it. Expanding $f(x,\varepsilon)$ in a Taylor series near 
$\varepsilon =0$ gives:
\begin{equation}\label{expand}
\dot{x} = f(x,0) + \cfrac{\partial f(x,\varepsilon)}{\partial \varepsilon}\Bigg|_{\varepsilon =0}\cdot \varepsilon + \mathcal{O}(\varepsilon^2) =: f_0(x) + \varepsilon f_1(x) + \mathcal{O}(\varepsilon^2).
\end{equation}
If $S_0$ is normally hyperbolic\footnote{Normal hyperbolicity of $S_0$ implies the algebraic and geometric
multiplicity of $D_1f(x,0)$'s trivial eigenvalue is equal to $k$.} and attracting, a splitting exists:
\begin{equation}\label{decomp}
T_x\mathbb{R}^n \cong \mathbb{R}^n = \ker D_1f(x,0) \oplus \;\text{image} \;D_1f(x,0)
\end{equation}
for all $x\in S_0$, where $T_xS_0 \cong \ker D_1f(x,0)$ for all $x\in S_0$. Projecting the right-hand side
of \eqref{expand} onto $T_xS_0$ approximates the dynamics on $S_0^{\varepsilon}$. This projection is
achieved using the operator $\pi^s:\mathbb{R}^n \to \ker D_1f(x,0)$; see~\citet{Goeke2014} 
and~\citet{Wechselberger2020} for details on the construction of $\pi^s$. The resulting $n$-dimensional dynamical 
system (with $k$ independent variables) is:
\begin{equation}\label{reduced}
x' =  \pi^s f_1(x)|_{x\in S_0}
\end{equation}
where ``$\phantom{x}'$" denotes differentiation with respect to the slow timescale, $\tau = \varepsilon t$. 
The dynamics on $S_0^{\varepsilon}$ determined by \eqref{slowflow} converges to \eqref{reduced} as $\varepsilon \to 0$. 
Equation~\eqref{reduced} is referred to as quasi-steady-state approximation (QSSA) in chemical kinetics.

\subsubsection{Convergence and Initial Conditions}
Solutions to \eqref{singular} with initial condition $x(0)=(z,y)(0)=(z_0,y_0)$ generally do not converge
to the solution of \eqref{reduced} with initial condition $x(0)$ unless $x(0)$ is 
sufficiently close to $S_0$. To ensure convergence, both \eqref{mmMA} and \eqref{reduced} generally 
require a modified initial ``$\tilde{x}_0$" (see, ~\ref{sec:appendix_initialdata} for details).
\subsection{From Theory to Practice: Tikhonov-Fenichel Parameter Values}
Traditional methods for putting mass action equations into perturbation form involve scaling and dimensional 
analysis, which can be unreliable. A more robust approach considers the mass action equations for a general 
reaction mechanism:
\begin{equation}
\dot{x} = f(x,p), \quad f:\mathbb{R}^n \times \mathbb{R}^m \to \mathbb{R}^n
\end{equation}
where $p$ represents an $m$-tuple of parameters. A point $p^*$ in parameter space is a Tikhonov-Fenichel 
Parameter Value (TFPV) \cite{Goeke2015, Goeke2017} if:
\begin{enumerate}
\item The set $S_0=\{x\in \mathbb{R}^n : f(x,p^*)=0\}$ is a $k$-dimensional submanifold of $\mathbb{R}^n$ with $0<k<n$, and
\item $S_0$ is normally hyperbolic.
\end{enumerate}
Expanding $f(x,p)$ near $p^*$ in a direction $\rho$, $p=p^* +\varepsilon \rho$, gives\footnote{The operator 
$D_2$ denotes differentiation with respect to the second argument, $p$.}
\begin{equation}
\dot{x} = f(x,p^*) + \varepsilon D_2f(x,p^*)\rho,
\end{equation}
which matches the form \eqref{reduced}. This highlights the need to specify a path in parameter space for 
taking the limit as $\varepsilon \to 0$ and $p\to p^*$.

For the MM system~\eqref{mmMA}, the parameter space is a subset of $\mathbb{R}^4_{\geq 0}$ with 
$p=(e_0,k_1,k_{-1},k_2)^T$. Here, $p^*:=(0,k_1,k_{-1},k_2)^T$ is a TFPV if $k_1$, $k_{-1}$, $k_2$ are bounded 
by positive constants. When $p=p^*$, the critical manifold, $S_0$, is the $s$-axis, which is normally 
hyperbolic and attractive: 
\begin{equation}
    S_0 := \{(s,c)\in \mathbb{R}^2_{\geq 0}: c=0\}.
\end{equation}

To put \eqref{mmMA} into perturbation form, we define a curve $\ell(\cdot)$ that passes through $p^*$
\begin{equation}
\ell(\varepsilon) = \begin{pmatrix}0\\k_1\\k_{-1}\\k_2\end{pmatrix} +  \varepsilon \begin{pmatrix}\widehat{e}_0\\0\\0\\0\end{pmatrix} = p^* + \varepsilon \rho
\end{equation}
where $e_0\mapsto \varepsilon \widehat{e}_0$. This effectively equates ``$\varepsilon$" with 
the magnitude of $e_0$.\footnote{Here, $\widehat{e_0}$ is of unit 
magnitude and essentially denotes the unit of $e_0$. We adhere to a similar hat notation for the rest 
of the work.} The perturbation form of \eqref{mmMA} for small $e_0$ is 
\begin{equation}\label{MMperturb}
\begin{pmatrix}\dot{s}\\\dot{c}\end{pmatrix} = \begin{pmatrix}(k_1s+k_{-1})c\\
-(k_1s+k_{-1}+k_2)c\end{pmatrix} + \varepsilon \widehat{e}_0k_1 s\begin{pmatrix}-1\\\;\;\;1\end{pmatrix}
\end{equation}
Projecting the right hand side of \eqref{MMperturb} onto $TS_0$ yields the sQSSA \eqref{sQSSA}:
\begin{equation}\label{sQSSA2}
\begin{pmatrix}\dot{s} \\ \dot{c}\end{pmatrix} = -\cfrac{\varepsilon \widehat{e}_0 k_2 s}{K_M+s}\begin{pmatrix}1\\0\end{pmatrix} \equiv
 -\cfrac{k_2e_0 s}{K_M+s}\begin{pmatrix}1\\0\end{pmatrix} .
\end{equation}

This derivation implies that the flow on $S_0^{\varepsilon}$ converges to \eqref{sQSSA2} as $e_0 \to 0$, 
with other parameters fixed and positive. Furthermore, the solution to the first component of \eqref{mmMA} also converges 
to the first component of \eqref{sQSSA2} since the initial condition $(s_0,0)\in S_0$.

Other TFPV values exist for \eqref{mmMA}. For instance, with $p^* = (e_0,k_1,k_{-1},0)^T$, the critical manifold becomes
\begin{equation}
    S_0 := \left\{(s,c)\in \mathbb{R}^2_{\geq 0}: c = \frac{e_0s}{K_S + s}\right\}
\end{equation} and the corresponding Fenichel 
reduction is the slow product QSS reduction:
\begin{equation}\label{slowp}
\begin{pmatrix}
\dot{s}\\\dot{c}
\end{pmatrix} = -\cfrac{k_2e_0 s(K_S+s)}{K_Se_0 + (K_S+s)^2}\begin{pmatrix} 1\\ \cfrac{K_Se_0 - K(K_S+s)}{(K_S+s)^2}\end{pmatrix}.
\end{equation}
This reduction differs from \eqref{sQSSA2} and will be analyzed in Section~\ref{sec:trapping}.

\subsection{The Open Problem: Identifying the Predominant Reduction}
Just as important as what Fenichel and TFPV theory tell us is what they do \textit{not}. While the theory 
indicates that \eqref{sQSSA2} approximates \eqref{mmMA} as $e_0\to 0$, it does not specify how small $e_0$ 
must be for \eqref{sQSSA2} to be valid in practice, where the singular limit is never truly reached. In other 
words, how small is small enough?

The size of $e_0$ must be considered relative to other parameters. Fenichel theory requires eigenvalue 
disparity near the critical manifold. For small $e_0$, the ratio of slow and fast eigenvalues near 
$(s,c)=(0,0)$ is: 
\begin{equation}\label{eig}
\delta :=\cfrac{e_0}{K_M}\cdot\cfrac{k_2}{k_{-1}+k_2} + \mathcal{O}(e_0^2).
\end{equation}
While $\delta \ll 1$ is necessary, it is not sufficient for the accuracy of \eqref{sQSSA2} \cite{eilertsen2023natural}. 
For instance, consider the case where $e_0=K_M$ but $k_2\ll k_{-1}$. This scenario renders small $\delta$, but the 
tangent vector to the graph of the $c$-nullcline at the origin will not align with the slow eigenvector, which is 
crucial for long-time accuracy. This misalignment can lead to significant errors in the sQSSA approximation, even 
when $\delta$ appears to indicate otherwise. 

The more restrictive condition $e_0\ll K_M$, introduced in \cite{reich1974mathematical}, has been rigorously 
proven sufficient to guarantee the accuracy of the 
sQSSA \cite{EILERTSEN2020, eilertsen2024rigorous, KangWon2017, Mastny2007, eilertsen2022stochastic}. 
However, this condition does not fully address the issue of predominance. For instance, if $k_{-1}$ is large 
and $e_0$ is reduced such that $k_1e_0 \ll k_{-1}$ (implying $e_0\ll K_M$), convergence to \eqref{sQSSA2}
is not implied. In fact, as $k_{-1} \to \infty$, the resulting singular perturbation problem yields a trivial 
QSSA: $\dot{s}=\dot{c}=0$ (see \ref{sec:largekminus1} for the formal calculation).

\begin{figure}
    \centering
     \begin{subfigure}
         \centering
         \includegraphics[scale = 0.38]{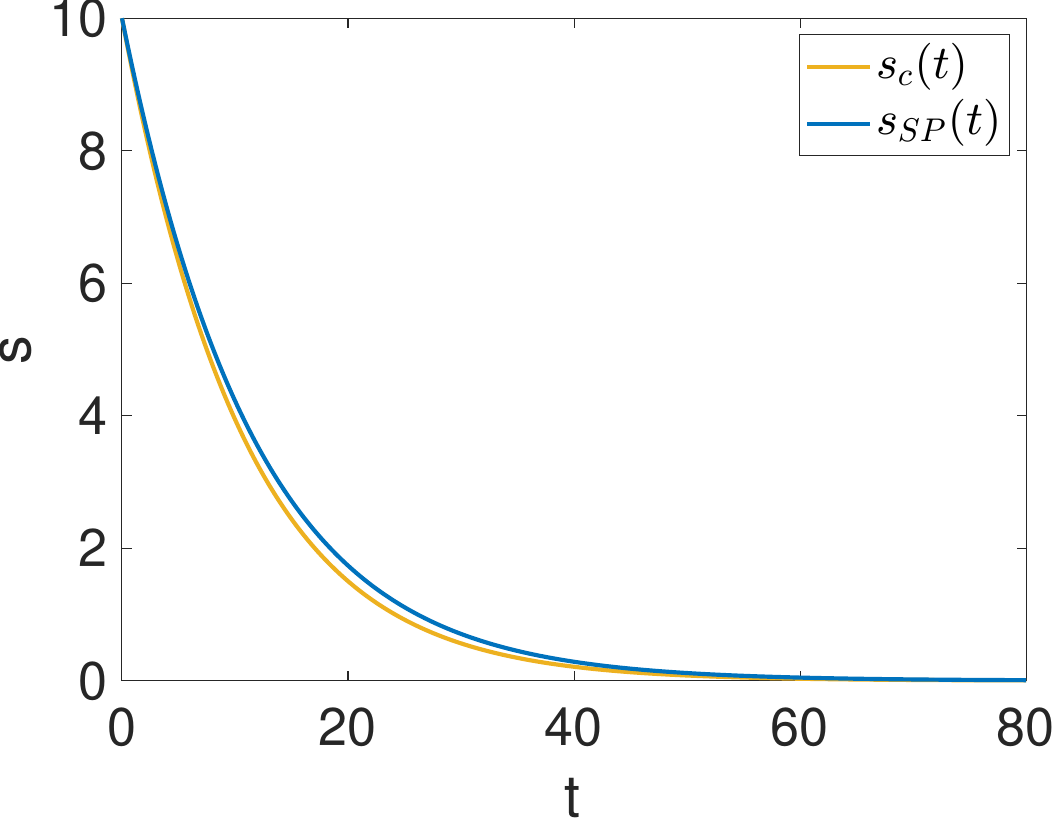}
     \end{subfigure}
     \qquad
     \begin{subfigure}
         \centering
         \includegraphics[scale = 0.38]{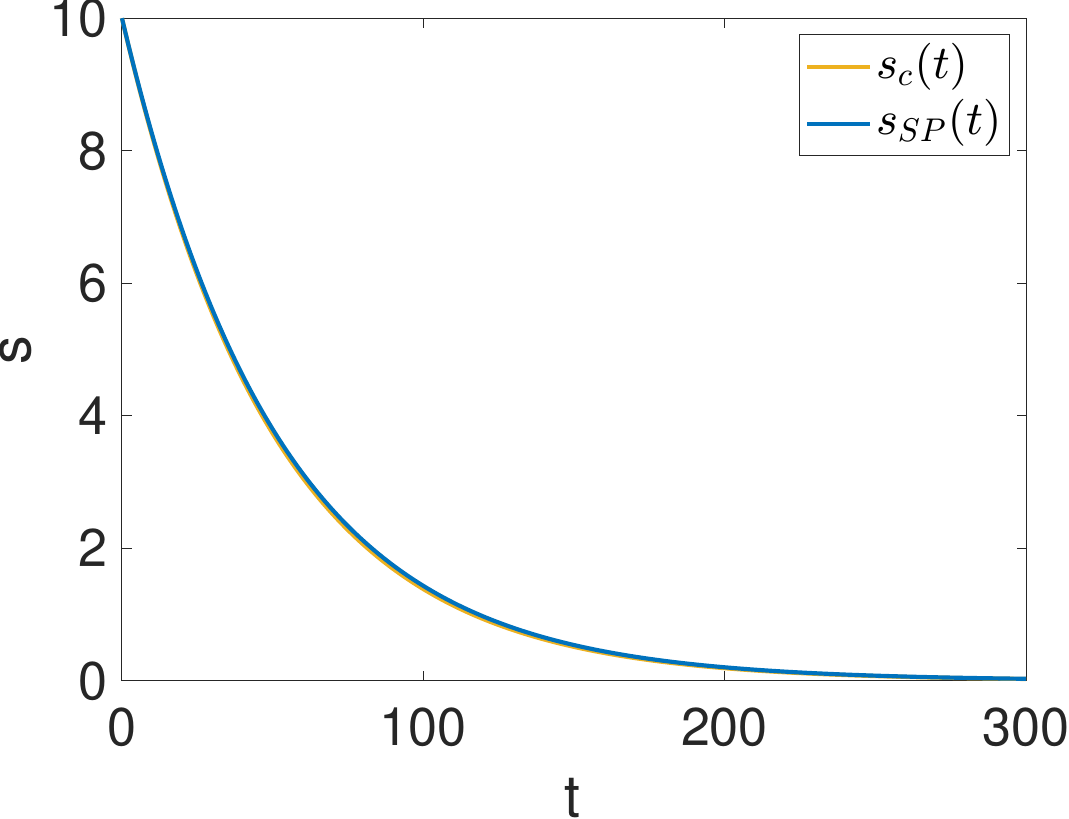}
     \end{subfigure}  
    \caption{\textbf{The sQSSA~\eqref{sQSSA2} and the slow product QSSA~\eqref{slowp} are 
    indistinguishable for large $k_{-1}$.} \textbf{Left}: $s_0 = 10$, $c_0 = 0$, $e_0 = 10$, $k_1 = 1$, 
    $k_{-1} = 100$, $k_2 = 1.$ \textbf{Right}: $s_0 = 10$, $c_0 = 0$, $e_0 = 10$, $k_1 = 1$, 
    $k_{-1} = 500$, $k_2 = 1.$ Substrate depletion over time is shown in both panels. $s_c(t)$ is 
    the numerical solution to \eqref{sQSSA2} and $s_{SP}(t)$ is the numerical solution to \eqref{slowp}.}
    \label{fig:largeKinverse}
\end{figure}

This example demonstrates that different TFPVs can share a critical manifold and be close in parameter 
space. For instance, consider a scenario where $k_1e_0$ and $k_2$ are relatively small compared to $k_{-1}$.
In this case, the TFPVs $(0,k_1,k_{-1},k_2)^T$ or $(e_0,k_1,k_{-1},0)^T$ could be quite close to each other. 
Terminating the approach to one TFPV might inadvertently lead to a point in parameter space that is 
actually closer to another TFPV.

Numerical simulations ({\sc Figure}~\ref{fig:largeKinverse}) illustrate this phenomenon. When $k_{-1}$ is 
significantly large, the trajectories of the system under the sQSSA~\eqref{sQSSA2} and the slow product 
QSSA~\eqref{slowp} become virtually indistinguishable. This makes it challenging to definitively determine 
which reduction is more accurate or appropriate for the given parameter values.

Moreover, the presence of multiple TFPVs introduces ambiguity in assessing the validity of the sQSSA. For 
example, if $e_0=1$, $k_1=0.1$, $k_{-1}=10$, and $k_2=0.1$, it is unclear whether the system's behavior
is better approximated by the sQSSA associated with the TFPV $(0,k_1,k_{-1},k_2)^T$ or the slow 
product QSSA associated with $(e_0,k_1,k_{-1},0)^T$.  This ambiguity arises because the parameter values 
do not fall neatly within the domain of a single TFPV, making it difficult to determine which reduction 
provides the most accurate representation of the system's dynamics. Thus, we define the 
\textit{predominance} of a given QSS reduction within the class of QSS approximations derived from Fenichel 
theory for nearby TFPVs. Established validity conditions, such as the Reich-Sel'kov condition, can often 
guide us toward identifying TFPVs with overlapping validity conditions, facilitating a robust analysis 
of predominance. However, in general, determining proximity to different TFPVs remains challenging due 
to the varying units of the parameters, necessitating a dimensionless indicator.

\textbf{Remark}: The approximation of the slow manifold is a well-studied area, with 
methods extending beyond Fenichel's reduction theory. For instance, the Fraser iterative method provides 
approximations for slow manifolds with small curvatures \cite{fraser1988steady, Nguyen1989, kaper2002asymptotic}. 
In such cases, incorporating additional iterative terms can enhance the accuracy of the approximations. 
Furthermore, several studies employ integration-based methods to generate numerical approximations of 
the slow flow within a specified error tolerance \cite{Davis1999}. In our work, we focus on other 
Fenichel reductions to examine the predominance of the standard QSSA, given that their convergence to 
the slow flow has been extensively established for the Michaelis-Menten mechanism.

The fact that the magnitude of $e_0/K_M$ alone does not ensure the sQSSA \eqref{sQSSA2} is the 
\textit{predominant} reduction is a consequence of having to always operate away from the singular limit 
in a large parameter space. This naturally leads to the question: does there exist a dimensionless 
parameter whose magnitude not only ensures the accuracy of the sQSSA \eqref{sQSSA} but also its predominance 
among other known QSSAs? In other words, can we define a more comprehensive and robust notion of 
``validity" that encompasses both accuracy and predominance?

\section{Anti-funnels and the Slow Invariant Manifold}\label{sec:preliminaries}
As discussed in Section~\ref{SEC2}, a central challenge in applying Fenichel theory to the MM system 
is the potential ambiguity arising from different TFPVs sharing the same critical manifold. This can 
lead to uncertainty about the validity of a specific reduction away from the singular limit. 
To address this, we turn to qualitative methods that allow us to determine the location of the slow 
manifold, $S_0^{\varepsilon}$, relative to known curves in the phase-plane. This section introduces 
the anti-funnel theorem, adapted from \citet{hubbard1997differential} and following the approach 
of \citet{calder2008properties}. We outline a strategy to determine when the sQSSA \eqref{sQSSA} is 
not only accurate but also predominant.

\begin{definition}\label{def:fences}
\textbf{Fences and anti-funnels.} Consider the first-order differential equation $y' = f(x,y)$ over 
the interval $x \in I = [a,b)$ where $a<b \leq \infty$ and ``$\phantom{x}'$" denotes differentiation 
with respect to $x$. Let $\alpha$ and $\beta$ be continuously-differentiable functions that satisfy 
\begin{equation}\label{funnel}
\alpha'(x) \leq f(x,\alpha(x)) \; \text{and} \; f(x,\beta(x)) \leq \beta'(x)
\end{equation}
for all $x \in I$.
\begin{itemize}
\item [(a)] The curve $\alpha$ is a \textit{lower fence} and the curve $\beta$ is an \textit{upper fence}. 
\item [(b)] The lower and upper fences are \textit{strong fences} if the respective inequalities are always strict. 
\item [(c)] The set 
\begin{equation}
\Gamma := \{(x,y) \; : \; x \in I, \; \beta(x) \leq y \leq \alpha(x)\}
\end{equation}
is called an \textit{anti-funnel}\footnote{In the $(x,y)$ phase plane, if the upper fence lies above 
the lower fence, the set between the two fences is known as a \textit{funnel}. Conversely, if the lower 
fence lies above the upper fence, the set between the two fences is called an 
\textit{antifunnel} \cite{hubbard1997differential}. Refer to \ref{sec:appendix_synopsis} for a brief 
synopsis on the theory of fences and anti-funnels and how it applies to the MM mechanism. 
See \citet{hubbard1997differential} for the detailed theory.} if  
$$\beta(x) < \alpha(x) \; \; \text{for all} \; \; x \in I.$$ 
\item [(d)] The anti-funnel is \textit{narrowing} if 
\begin{equation}
    \lim\limits_{x\rightarrow b^-} |\alpha(x) - \beta(x)| = 0.
\end{equation}    
\end{itemize}
\end{definition}

\begin{theorem}\label{thm:antifunnel}
\textbf{Anti-funnel Theorem} \cite{hubbard1997differential}. Consider the first-order differential 
equation $y' = f(x,y)$ over the interval $I = [a,b)$ where $a<b \leq \infty$. If $\Gamma$ is an 
\textit{anti-funnel} with a strong lower fence, $\alpha$, and a strong upper fence, $\beta$, then 
there exists a solution $y(x)$ to the differential equation such that 
\begin{equation}\label{fences}
    \beta(x) < y(x) < \alpha(x) \; \; \text{for all} \; \; x \in I.
\end{equation}
Furthermore, if $\Gamma$ is \textit{narrowing} and $\frac{\partial f}{\partial y} (x,y) \geq 0$ 
in $\Gamma$, then the solution $y(x)$ is unique. 
\end{theorem}
\begin{figure}[tb!]
    \centering
    \includegraphics[scale=0.45]{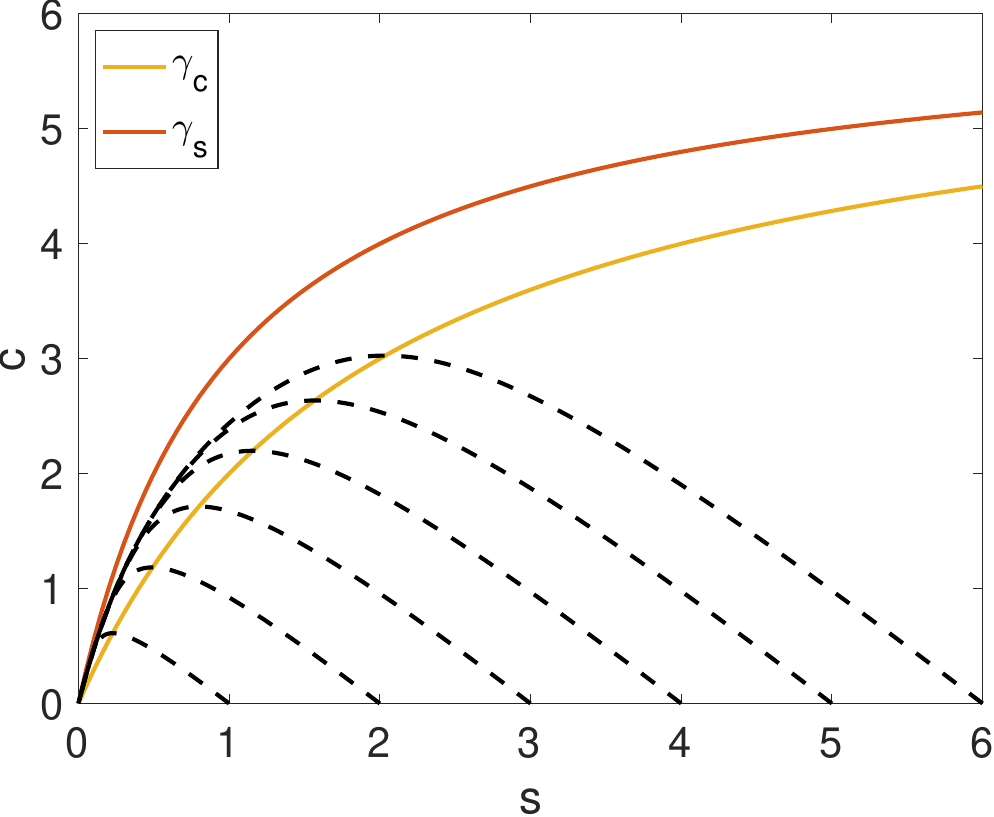}
    \caption{\textbf{The region, $\Gamma$, between the horizontal nullcline and the vertical 
    nullcline curves is a trapping region on the phase plane.} The numerical solutions to the 
    mass action equations~\eqref{MM} for several initial conditions are denoted by black dashed 
    lines. $s_0 \in \{1,2,3,4,5,6\}$, $c_0 = 0$, $e_0 = 6$, $k_1 = 1$, $k_{-1} = 1$, $k_2 = 1.$
    The numerical solution of trajectories start on the $s$-axis and eventually enter $\Gamma$.}
    \label{fig:isoclines_solutions}
\end{figure} 

For the MM equations~\eqref{mmMA}, we have:
\begin{equation}\label{1D}
\frac{\mathrm{d}c}{\mathrm{d}s} = -\cfrac{e_0s -(K_M+s)c}{e_0s-(K_S+s)c} =: f(c,s).
\end{equation}
The distinguished solution to \eqref{1D}, ``$c=y(s)$", satisfying \eqref{fences} represents the 
slow manifold, $S_0^{\varepsilon}$. The challenge lies in finding suitable lower and upper fences, 
$\alpha(s)$ and $\beta(s)$. 

A natural choice for the upper fence is the $c$-nullcline, $\gamma_c$, given by:
\begin{equation}\label{graphC}
\text{Graph}(\gamma_c)=\{(s,c) \in \mathbb{R}^2: c = e_0s/(K_M+s)\}
\end{equation}
Considering the biochemically relevant portion of $\mathbb{R}^2$, we focus on 
$\text{Graph}(\gamma_c) \cap \mathbb{R}^2_{\geq 0}$. This choice is advantageous because it 
represents the quasi-steady-state variety corresponding to the standard reduction~\eqref{sQSSA}, 
and all phase-plane trajectories starting on the s-axis eventually cross it (see 
\citet{calder2008properties} and \citet{Noethen2007} for a proof of this statement).

The set $\Gamma$ enclosed between the $c$-nullcline and the $s$-nullcline, $\gamma_s$,\footnote{The 
graph of the $s$-nullcline is the set $\{(s,c) \in \mathbb{R}^2: c = e_0s/(K_S+s)\},$} is positively 
invariant:
\begin{equation}
\Gamma :=\{(s,c)\in \mathbb{R}^2_{\geq 0} \; : \; \gamma_c(s) \leq c \leq \gamma_s(s)\}
\end{equation}
Once a trajectory enters $\Gamma$, it remains within $\Gamma$ as $t\to \infty$. (see \citet{calder2008properties} 
and \citet{Noethen2007} for a detailed proof). {\sc Figure}~\ref{fig:isoclines_solutions} illustrates 
this behavior.

Constructing a suitable lower fence requires more analysis. The stationary point $(s,c)=(0,0)$ is a node
with eigenvalues $\lambda_-<\lambda_+<0$:
\begin{equation}
\lambda_{\pm} = \cfrac{k_1}{2}(K_M+e_0)\left(-1\pm\sqrt{1-\cfrac{4Ke_0}{(K_M+e_0)^2}}\right).
\end{equation}
Under timescale separation in which
\begin{equation}
\cfrac{Ke_0}{(K_M+e_0)^2} \ll 1/4,
\end{equation}
trajectories eventually approach the stationary point along the slow eigenvector, $v^+$, spanning the one-dimensional 
subspace:
\begin{equation}\label{TangentO}
T_0S_0^{\varepsilon} = \text{span}(v^+) = \{(s,c)\in \mathbb{R}^2: c= ms\}, 
\end{equation}
with slope $m$:

\begin{equation}\label{eq:alpha}
m= \frac{1}{2k_{-1}}\left(-k_{-1} - k_2 + k_1e_0 + \sqrt{(k_{-1} + k_2 + k_1e_0)^2 - 4k_1k_2e_0}\right).
\end{equation}
\citet{calder2008properties} leverage this property and define the lower fence:
\begin{equation}\label{CSfence}
\alpha_{cs}(s) := \cfrac{me_0s}{e_0+ms},
\end{equation}
and prove that $S_0^{\varepsilon}$ lies between $\gamma_c(s)$ and $\alpha_{cs}(s)$. Notably, the graph 
of $\alpha_{cs}$ lies below that of $\gamma_s$. By proving that
\begin{equation}
\Gamma_{cs} := \{(s, c) \; : \; s\in I, \gamma_c(s) \leq c \leq  \alpha_{cs}(s)\}
\end{equation}
is a narrowing anti-funnel,\footnote{\citet{calder2008properties} define $I$ as $[a, \infty)$ where 
$a > 0$ is arbitrary.} \citet{calder2008properties} prove that the distinguished slow manifold, 
$S_0^{\varepsilon}$, also lies in $\Gamma$, the region contained between the $c$- and $s$-nullcline. 

A key feature of \eqref{CSfence} is $\alpha_{cs}'(0)=m$. This leads to highly accurate reduced equation 
for $\dot{s}$ near the stationary point:\footnote{The equation is obtained via direct substitution of 
$c=\alpha_{cs}(s)$ into \eqref{mms}.}
\begin{equation}\label{CSR}
\dot{s} = -\cfrac{e_0s}{e_0+ms} \cdot (k_1e_0 - k_{-1}m).
\end{equation}
Despite its advantages, this approach has limitations. First, the reduction~\eqref{CSR} may be unreliable 
away from the stationary point, even when $\lambda_- \ll \lambda_+ < 0$ (see ~\citet{eilertsen2022anti}). 
Second, the anti-funnel construction leads to the following:
\begin{proposition}
For \eqref{mmMA} with initial condition $(s,c)(0)=(s_0,0)$, let $t=t_{\rm cross}$ be the time the 
trajectory crosses the upper fence $c=\gamma_c(s)$. Then, the bound
\begin{equation}\label{bound1}
-\cfrac{k_2e_0s}{K_M+s} \leq \dot{s} \leq -\cfrac{e_0s}{e_0+ms} \cdot (k_1e_0 - k_{-1}m)
\end{equation}
holds for all $t\geq t_{\rm cross}$. 
\end{proposition}
\begin{proof}
Since $\Gamma_{cs}$ is a narrowing anti-funnel, $S_0^{\varepsilon}$ lies between the graphs of 
$\alpha_{cs}(s)$ and $\gamma_c(s)$. As $S_0^{\varepsilon}$ is invariant, any trajectory entering 
$\Gamma_{cs}$ cannot cross it. Therefore, $\beta(s) \leq c \leq \alpha_{cs}(s)$ holds for all 
$t\geq t_{\rm cross}$ since $\Gamma$ is positively invariant and $\Gamma_{cs} \subset \Gamma$. 
With $\dot{s} = -k_1e_0s + k_1(K_S+s)c$, it follows that
\begin{equation*}
-k_1e_0s + k_1(K_S+s)\gamma_c(s) \leq \dot{s} \leq -k_1e_0s + k_1(K_S+s)\alpha_{cs}(s)
\end{equation*}
proving the assertion.
\end{proof}

While the upper bound in \eqref{bound1} is sharper than $\dot{s} \leq 0$, extracting quantitative 
information about the sQSSA's accuracy and dimensionless parameters like $\varepsilon_{RS} = e_0/K_M$ 
from \eqref{bound1} is not straightforward.

The goal is to define an anti-funnel with fences that provide both qualitative and quantitative 
information about $S_0^{\varepsilon}$ and the error for a given QSSA. By ``trapping" 
$S_0^{\varepsilon}$ between suitable upper and lower fences that form narrowing anti-funnels, we 
can derive qualifiers that determine the accuracy and predominance of various QSSAs.

\section{Trapping the Slow Manifold: The Standard QSSA}\label{sec:trapping}
As stated in Section \ref{sec:preliminaries}, our strategy is to construct upper and lower fences 
that form a narrowing anti-funnel. This raises the question: How do we construct suitable fences? 
One approach is approximating the slow manifold, $S_0^{\varepsilon}$, via perturbation expansion
\begin{equation}\label{Texp}
S_0^{\varepsilon}= h_0(s)+\varepsilon h_1(s) + \varepsilon^2h_2(s) + \mathcal{O}(\varepsilon^2).
\end{equation}
Insertion of \eqref{Texp} into the invariance equation results in a \textit{regular} perturbation 
problem. However, the coefficients in \eqref{Texp} depend on $\varepsilon$, which in turn depends 
on a  specific TFPV. Since we aim for general error bounds, we need fences that are not heavily 
$\varepsilon$-dependent.

An alternative approach involves combining quantitative reasoning with creative insights. 
\citet{kumar2011reduced} sought to approximate the flow on $S_0^{\varepsilon}$ for small 
$k_2$ ($k_2 \mapsto \varepsilon \widehat{k}_2$). They reasoned that
\begin{equation}\label{ASM}
c=\cfrac{e_0s}{K_M+s}
\end{equation}
is a good approximation to $S_0^{\varepsilon}$ when $k_2$ is small. Under this assumption, the 
flow of $c$ on $S_0^{\varepsilon}$ is determined by:
\begin{equation}\label{cdotk2}
\dot{c} \approx \cfrac{{\rm d}}{{\rm d}s}\left(\cfrac{e_0s}{K_M+s}\right)\dot{s} = \cfrac{K_Me_0}{(K_M+s)^2}\cdot \dot{s}.
\end{equation}
Substituting \eqref{cdotk2} and \eqref{ASM} into the conservation law, $\dot{s} + \dot{c} + k_2c=0$, 
and solving for $\dot{s}$ yields:
\begin{equation}\label{eq:improvedQSSreduction}
    \dot{s} = -\frac{k_2e_0s(K_M+s)}{K_Me_0 + (K_M+s)^2}.
\end{equation}

This reduction is significant for two reasons. First, it effectively captures the Fenichel 
reduction for this scenario (compare to \eqref{slowp} with small $k_2$). Second, the graph of $\gamma(s)$, where
\begin{equation}\label{eq:improvedQSS}
\gamma: s\mapsto \frac{e_0s}{K_S+s} - \frac{k_2e_0s(K_M+s)}{k_1(K_S+s)(K_Me_0 + (K_M+s)^2)},
\end{equation}
is an excellent candidate for an appropriate anti-funnel fence. 

\begin{figure}[tb!]
   \centering
   \includegraphics[scale=0.45]{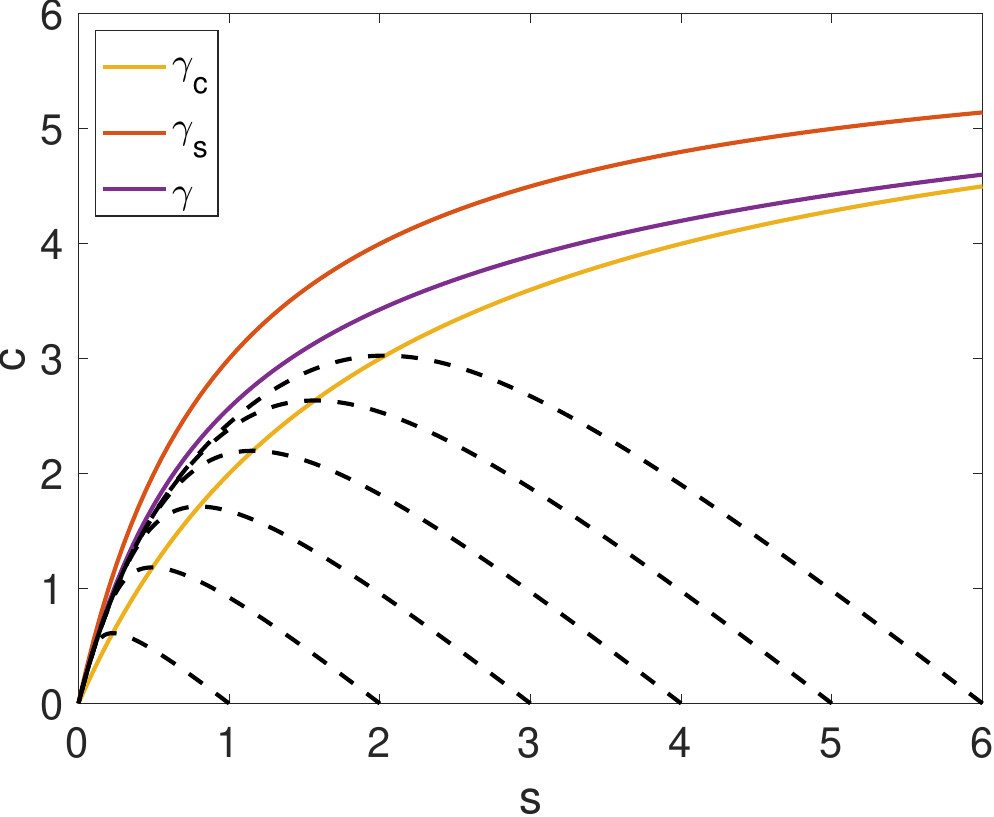}
   \caption{\textbf{The region between the horizontal nullcline curve $\gamma_c(s)$ and the 
   curve $\gamma(s)$ is a trapping region on the phase plane.} The numerical solutions to 
   the mass action equations~\eqref{MM} for several initial conditions are denoted by black 
   dashed lines. $s_0 \in \{1,2,3,4,5,6\}$, $c_0 = 0$, $e_0 = 6$, $k_1 = 1$, $k_{-1} = 1$, 
   $k_2 = 1.$}
    \label{fig:isoclines_improved_solutions}
\end{figure} 

As shown in Figure \ref{fig:isoclines_improved_solutions}, $\gamma_c, \gamma_s$ and $\gamma$ 
are strictly increasing, with $\gamma_c < \gamma < \gamma_s$ for all $s > 0$. The set contained 
between the graphs of $\gamma_c$ and $\gamma$ is positively invariant:
\begin{equation}\label{gamma}
    \Gamma_0 := \{(s,c)\in \mathbb{R}^2_{\geq 0} \; : \; \gamma_c(s) \leq c \leq \gamma(s)\}.
\end{equation}
Let $t_{\rm cross}$ be the time a trajectory crosses the $c$-nullcline, as established in
\cite{Noethen2007,calder2008properties}. Then, $(s,c)(t) \in \Gamma_0$ for all $t \geq t_{\rm cross}$.
Importantly:

\begin{theorem}\label{thm:trapping_gamma_0}
For the differential equation~\eqref{1D}, $\Gamma_0$ is a positively invariant, narrowing 
anti-funnel containing a unique slow manifold, $S_0^{\varepsilon}$.
\end{theorem}
\begin{proof}
To prove $\Gamma_0$ is positively invariant, we show that the vector field for \eqref{mmMA} 
points inwards at the boundary curves:
\begin{subequations}\label{eq:vectorfieldineq}
    \begin{align}
        &\dot{c}(s,\gamma(s)) \;- \gamma'(s)\dot{s}(s,\gamma(s)) \;< 0 \label{eq:ineq1}, \\
        &\dot{c}(s,\gamma_c(s)) - \gamma_c'(s)\dot{s}(s,\gamma_c(s)) > 0 \label{eq:ineq2}.
    \end{align}
\end{subequations}
See \ref{sec:appendix_trapping1} for details. Positive invariance implies trajectories entering
$\Gamma_0$ remain within. \citet{calder2008properties}, as well as~\citet{Noethen2007}, prove that for solutions initializing on 
the $s$-axis, there exists a $t_{\rm cross} > 0$ at which they enter $\Gamma_0$ by crossing $\gamma_c$.
    
From \eqref{eq:vectorfieldineq}, $\gamma$ and $\gamma_c$ are strong lower and upper fences, 
respectively, with $\gamma(s) > \gamma_c(s)$ for $s > 0$. Furthermore, 
\begin{equation}
\lim\limits_{s \rightarrow \infty} |\gamma(s) - \gamma_c(s)| = 0 \quad \text{and}  \quad \frac{\partial f}{\partial c} (c,s) \geq 0
\end{equation} 
in $\Gamma_0$. The assertion follows from 
Theorem ~\ref{thm:antifunnel}. 
\end{proof}

Theorem~\ref{thm:trapping_gamma_0} establishes that $\Gamma_0$ contains the slow manifold 
$S_0^{\varepsilon}$, independent of any specific perturbation scenario. This allows us to 
extract qualitative and quantitative information about the accuracy of a QSSA and assess 
predominance when multiple TFPVs are close.

The rest of this section is organized as follows. In Section~\ref{sec:trapping_sqss}, we use 
the improved trapping region to derive a quantitative error estimate for the sQSS approximation 
and recover the Reich-Sel'kov condition. In Section~\ref{sec:trapping_rQSS}, we explore the 
conditions where both the sQSS and reverse QSS reductions are valid, examining their 
validity as we move away from the Reich-Sel'kov condition.

\subsection{Accuracy of the Standard QSSA}\label{sec:trapping_sqss}
To leverage $\Gamma_0$'s properties, we introduce sharp bounds on the substrate depletion 
for $t \geq t_{\rm cross}$, representing the slow regime.
\begin{proposition}\label{prop:substratebound}
For system~\eqref{mmMA} with initial condition $(s, c)(0) = (s_0, 0)$, the bound
\begin{equation}\label{eq:substratebound}
    -\frac{k_2e_0s}{K_M + s} \leq \dot{s} \leq -\frac{k_2e_0s(K_M+s)}{K_Me_0 + (K_M+s)^2}
\end{equation}
holds for all $t \geq t_{\rm cross}$.
\end{proposition}
\begin{proof}
The proof follows from the construction of $\Gamma_0$ and substituting 
$\gamma_c(s) \leq c \leq \gamma(s)$ into \eqref{mms}.
\end{proof}

Since the LHS of \eqref{eq:substratebound} is the sQSSA for substrate depletion, the design 
of $\Gamma_0$ suggests that the sQSSA performs 
well when $\gamma$ and $\gamma_c$ are close. Rewriting \eqref{eq:substratebound} as: 
\begin{equation}\label{sbound}
       -\frac{k_2e_0s}{K_M+s} \leq \dot{s} \leq 
       -\frac{k_2e_0s}{K_M+s}\left(\frac{1}{1 + \varphi(s)}\right), \quad \varphi(s) = \cfrac{K_Me_0}{(K_M+s)^2}
\end{equation}
we see that, for any subinterval $\ell \subset [0,s_0]$, 
\begin{equation}
-\frac{k_2e_0s}{K_M+s} \leq \dot{s} \leq 
       -\frac{k_2e_0s}{K_M+s}\left(\frac{1}{1 + \displaystyle \max_{s\in \ell}\varphi(s)}\right),\quad t_{\rm cross} \leq t.
\end{equation}
Since
\begin{equation}
\displaystyle \max_{0\leq s \leq s_0}\varphi(s) = \cfrac{e_0}{K_M} = \varepsilon_{RS},
\end{equation}
we have:
\begin{equation}\label{EST}
-\frac{k_2e_0s}{K_M+s} \leq \dot{s} \leq 
       -\frac{k_2e_0s}{K_M+s}\cdot \sum_{j=0}^{\infty}(-1)^j\varepsilon_{RS}^j, \quad t_{\rm cross}\leq t, \quad \varepsilon_{RS} < 1.
\end{equation}

This analysis has several significant consequences. First, we have established the sQSSA's asymptotic 
dependency on $\varepsilon_{RS}$ without relying on non-dimensionalization, unlike 
traditional methods \cite[see][among others]{Heineken1967,reich1974mathematical,Segel1989}. 
This is advantageous due to the non-uniqueness of non-dimensionalization, which can lead to 
different ``small parameters" \cite{Heineken1967, reich1974mathematical, Segel1988}. Second, 
the bound \eqref{EST} is more aesthetic than the estimate in \citet{eilertsen2024rigorous}:
\begin{equation}\label{ESW}
-\frac{k_2e_0s}{K_M+s} \leq \dot{s} \leq 
       -(1-m)\frac{k_2e_0s}{K_M+s}, \quad t_{\rm cross}\leq t,
\end{equation}
where ``$m$"  is the slope of the slow tangent vector, $v^+$. 

Third, and this is somewhat unexpected, we gain more information about 
the approximation error of the flow on $S_0^{\varepsilon}$ from $\varphi(s)$ than from $m$.
From \eqref{sbound}, the sQSSA approximates the flow on $S_0^{\varepsilon}$ well in the
regions where:
\begin{equation}
K_Me_0 \ll (K_M+s)^2,
\end{equation}
even if $\varepsilon_{RS} \sim 1$. In particular, the sQSSA is a good approximation in 
regions where:
\begin{equation}
\max\{K_M,e_0\} \ll s.
\end{equation}
This relates to the validity of the reverse quasi-steady-state approximation (rQSSA), which 
operates in high enzyme concentrations. To revisit the sQSSA's predominance, we examine the 
rQSSA's validity and investigate any overlap in their conditions.

\subsection{Validity of the sQSSA for High Enzyme Concentrations}\label{sec:trapping_rQSS}
Equation \eqref{sbound} implies that the sQSSA is an excellent approximation to the flow on 
$S_0^{\varepsilon}$ when $K_M \ll s$ and $e_0 \lesssim s$. In this region, the graphs of 
the $c$-nullcline and $\gamma(s)$ approach their horizontal asymptote, $c=e_0$, reflected 
in $\varphi(s) \to 0$ as $s\to \infty$. The slow manifold, $S_0^{\varepsilon}$, is nearly
horizontal implying $\dot{c} \approx 0$ on $S_0^{\varepsilon}$ even for larger 
$\varepsilon_{RS}$.  

Two singular perturbation scenarios lead to a nearly horizontal slow manifold: the sQSSA, 
where the slow manifold coalesces with the $s$-axis as $k_1e_0 \to 0$, and the rQSSA, 
coinciding with small $k_{-1}$ and small $k_{2}$:
\begin{subequations}
\begin{align}
\dot{s} &= -k_1(e_0-c)s+\varepsilon \widehat{k}_{-1}c,\\
\dot{c} &= \;\;\;k_1(e_0-c)s -\varepsilon \widehat{k}_{-1} -\varepsilon \widehat{k}_2.
\end{align}
\end{subequations}
In the rQSSA scenario, the set of stationary points in the singular limit is not a 
submanifold of $\mathbb{R}^2$:
\begin{equation*}
S_0 = S_0^{(1)} \cup S_0^{(2)} 
\end{equation*}
where\footnote{The constants $\zeta_1,\zeta_2$ are introduced to enforce the compactness of $S_0^{(1)}$ and $S_0^{(2)}$.}
\begin{subequations}
\begin{align}
S_0^{(1)} &= \{(s,c)\in \mathbb{R}^2_{\geq 0}: s=0\;\&\;0 \leq c\leq e_0-\zeta_1\}, \quad 0 < \zeta_1 < e_0,\\
S_0^{(2)} &= \{(s,c)\in \mathbb{R}^2_{\geq 0}: c=e_0\;\&\;\zeta_2 \leq s \leq s_0\}, \;\qquad 0 < \zeta_2 < s_0.
\end{align}
\end{subequations}
Classical Fenichel theory does not apply to the entire set, but it applies to specific 
compact submanifolds. The resulting Fenichel reduction via projection onto $TS_0^{(2)}$ is:
\begin{subequations}\label{sred}
\begin{align}
\dot{s} &= -k_2e_0,\\
\dot{c} &= 0.
\end{align}
\end{subequations}
Likewise, the Fenichel reduction obtained via projection onto $TS_0^{(1)}$ is:
\begin{subequations}\label{cred}
\begin{align}
\dot{s} &= 0,\\
\dot{c} &= -k_2c.
\end{align}
\end{subequations}

In the rQSSA, trajectories initially approach the line $c=e_0$ and stay close until 
reaching the vicinity of the transcritical bifurcation point, $(s,c)=(0,e_0)$. Near 
this point, trajectories approach $(0,0)$ as $t\to \infty$; however, the slow eigenvector 
in this scenario is nearly indistinguishable from the $c$-axis (and in fact aligns 
with the $c$-axis in the singular limit). 

The rQSSA's long-time validity requires $\varepsilon_{RS}^{-1} \ll 1$, while the 
sQSSA's long-time validity requires $\varepsilon_{RS} \ll 1$~\cite{EILERTSEN2020}. 
However, through comparison, the sQSSA with \eqref{sred} reveals that they are practically 
indistinguishable when $K_M \ll e_0 \lesssim s$. Thus, the sQSSA can approximate the 
flow on $S_0^{\varepsilon}$ to the right of the bifurcation point, extending its 
validity beyond the Reich-Sel'kov parameter.

A natural question is: If we use the sQSSA to approximate the flow on $S_0^{\varepsilon}$ 
to the right of the  bifurcation instead of \eqref{sred}, how reliable is it in an 
asymptotic sense? Since $S_0^{\varepsilon}$ lies within $\Gamma_0$ when $0<\varepsilon$, 
we can get a rough answer. Assuming $\varepsilon_{RS}^{-1} \ll 1$ and considering 
$s\in [e_0,s_0]$, it follows from \eqref{sbound} that, for $t\geq t_{\rm cross}$,
\begin{equation}\label{sbound2}
\begin{aligned}
-\cfrac{k_2e_0s}{K_M+s} \leq \dot{s} &\leq -\cfrac{k_2e_0s}{K_M+s}\left( \cfrac{1}{1+\displaystyle \max_{s\in [e_0,s_0]}\varphi(s)}\right)\\
&= -\cfrac{k_2e_0s}{K_M+s}\left( \cfrac{1}{1+\displaystyle \mu}\right), \quad \mu:=\cfrac{\varepsilon_{RS}}{(1+\varepsilon_{RS})^2} < \varepsilon_{RS}^{-1}\\
&\leq -\cfrac{k_2e_0s}{K_M+s}\left( \cfrac{1}{1+\varepsilon_{RS}^{-1}}\right)\\
&=-\cfrac{k_2e_0s}{K_M+s}\cdot \sum_{j=0}^{\infty}(-1)^j\varepsilon_{RS}^{-j}, \quad 1<\varepsilon_{RS}.
\end{aligned}
\end{equation}
The bound \eqref{sbound2} implies the sQSSA is a good approximation when $\varepsilon_{RS}^{-1}\ll 1$ and 
$e_0 \lesssim s$, improving as $e_0/s \to 0$.

While the sQSSA approximates the flow on $S_0^{\varepsilon}$ well when 
$\varepsilon_{RS}^{-1} \ll 1$ and $e_0\lesssim s$, it cannot be equipped with $s_0$ 
as the initial substrate concentration in the rQSSA scenario. If the initial condition 
is $(s,c)(0)=(s_0,0)$, with $e_0 < s_0$, Fenichel theory dictates that the 
reduction~\eqref{sred} should be equipped with $(s,c)=(s_0-e_0,e_0).$  This extends 
to the sQSSA if used to approximate the flow on $S_0^{\varepsilon}$ to the right of 
the bifurcation point when $\varepsilon_{RS}^{-1}\ll 1$. 

The existence of multiple accurate QSS reductions under the same conditions necessitates 
evaluating which QSSA is predominant. To truly verify the standard-QSS reduction's 
applicability, we must explore other QSS reductions near the sQSS curve in the phase-plane. 
As noted in Section~\ref{sec:preliminaries}, in scenarios like large $k_{-1}$, the 
sQSSA~\eqref{sQSSA2} and the slow product QSSA~\eqref{slowp} are indistinguishable. 
To find the most accurate reduced system and its validity conditions, we next 
investigate the slow manifold's location relative to the (algebraic) variety that 
generates the slow product QSS reduction \eqref{slowp}.

\section{Trapping the Slow Manifold: The Slow Product QSSA}\label{sec:trapping_spQSS}
This section defines a new upper fence and anti-funnel for the slow manifold using the 
slow product QSS curve under specific parametric restrictions. We compare this to previous 
results to assess the approximation accuracy of the sQSSA \eqref{sQSSA2} and the slow product QSSA \eqref{slowp}, 
revisiting the Reich-Sel'kov condition and investigating whether it ensures the 
predominance of the sQSSA.

\begin{figure}[tb!]
    \centering
    \includegraphics[scale=0.6]{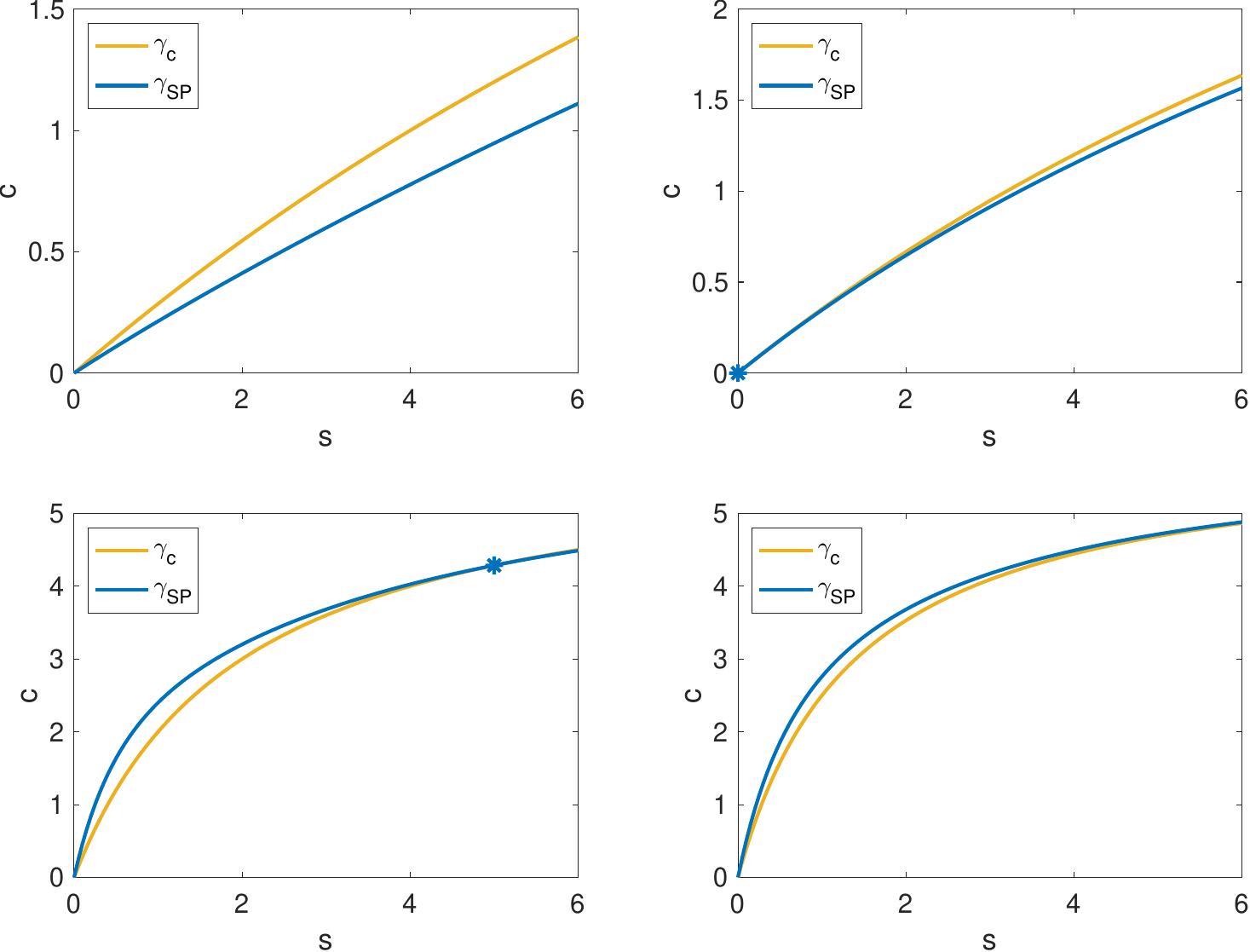}
    \caption{\textbf{The slow product QSS curve $\gamma_{SP}$ lies above the horizontal 
    nullcline $\gamma_c$ for $s < s^*$ \eqref{eq:sp_intersection} in the phase plane.} 
    \textbf{Top Left ($k_2 > e_0k_1$)}: $k_1 = 0.1$, $k_2 = 1$, $k_{-1} = 1$, $e_0 = 6$, 
    $s^* = -4$. \textbf{Top Right ($k_2 = e_0k_1$)}: $k_1 = 0.1$, $k_2 = 0.6$, $k_{-1} = 1$, 
    $e_0 = 6$, $s^* = 0$. \textbf{Bottom Left ($k_2 < e_0k_1$)}: $k_1 = 1$, $k_2 = 1$, 
    $k_{-1} = 1$, $e_0 = 6$, $s^* = 5$. \textbf{Bottom Right ($k_2 < e_0k_1$)}: $k_1 = 1$, 
    $k_2 = 0.4$, $k_{-1} = 1$, $e_0 = 6$, $s^* = 14$.}
    \label{fig:slow_product_movement_panel}
\end{figure} 

The slow product QSS reduction \eqref{slowp}, derived from Fenichel theory for small $k_2$,
corresponds to the QSS curve
\begin{equation}\label{eq:qss_sp}
    c = \gamma_{SP}(s) := \frac{e_0s}{K_S+s} - \frac{k_2e_0s}{k_1(K_Se_0 + (K_S+s)^2)}.
\end{equation}
This closely resembles the curve $\gamma(s)$. In Section~\ref{sec:trapping_sqss}, we 
established that the slow manifold lies between $\gamma(s)$ and $\gamma_c(s)$.
Now, we explore how the slow product QSS variety fits into these findings and the
insights its location provides concerning the sQSSA's validity and predominance.

In the phase plane, $\gamma_{SP}(s)$ lies above the horizontal nullcline $\gamma_c(s)$ 
for $s < s^*$ where
\begin{equation}\label{eq:sp_intersection}
    s^* = \frac{k_{-1}e_0}{k_2} - \frac{k_{-1}}{k_1} = \frac{k_{-1}}{k_2}(e_0 - K)
\end{equation}
is their intersection point for substrate concentration. We are primarily interested 
in cases where $s^*$ is positive, placing the intersection within the first quadrant. 
When $s^*$ is negative, or equivalently, when
\begin{equation}
    e_0 < K,
\end{equation}
the slow manifold crosses the sQSS variety, and the sQSSA is the best known reduction.

However, as $e_0$ increases, $s^*$ shifts to the right. {\sc Figure}~\ref{fig:slow_product_movement_panel} 
shows how the curves' positions change with parameter values. When $e_0 > K$, the 
intersection point lies in the first quadrant and $\gamma_{SP}(s) \geq \gamma_c(s)$ 
for $0 \leq s \leq s^*$. Also, $\gamma_{SP}(s) \leq \gamma(s)$ for all $s \geq 0$. Thus, 
the slow product QSS curve lies within $\Gamma_0$ for a significant portion of the 
phase plane. This raises the questions: Do solutions cross $\gamma_{SP}(s)$ after 
crossing $\gamma_c(s)$ under any parametric conditions? If so, can we locate the slow 
manifold more precisely? 

Numerical evidence suggests this is possible, with solutions lying closer to the slow 
product QSS curve in the steady-state regime (see Figure \ref{fig:slow_product_solutions}). 
Interestingly, the parametric conditions for this also ensure the slow-product QSS 
reduction's dominance over the sQSS reduction, aligning with our overall goal.


\begin{figure}[tb!]
    \centering
    \includegraphics[scale=0.45]{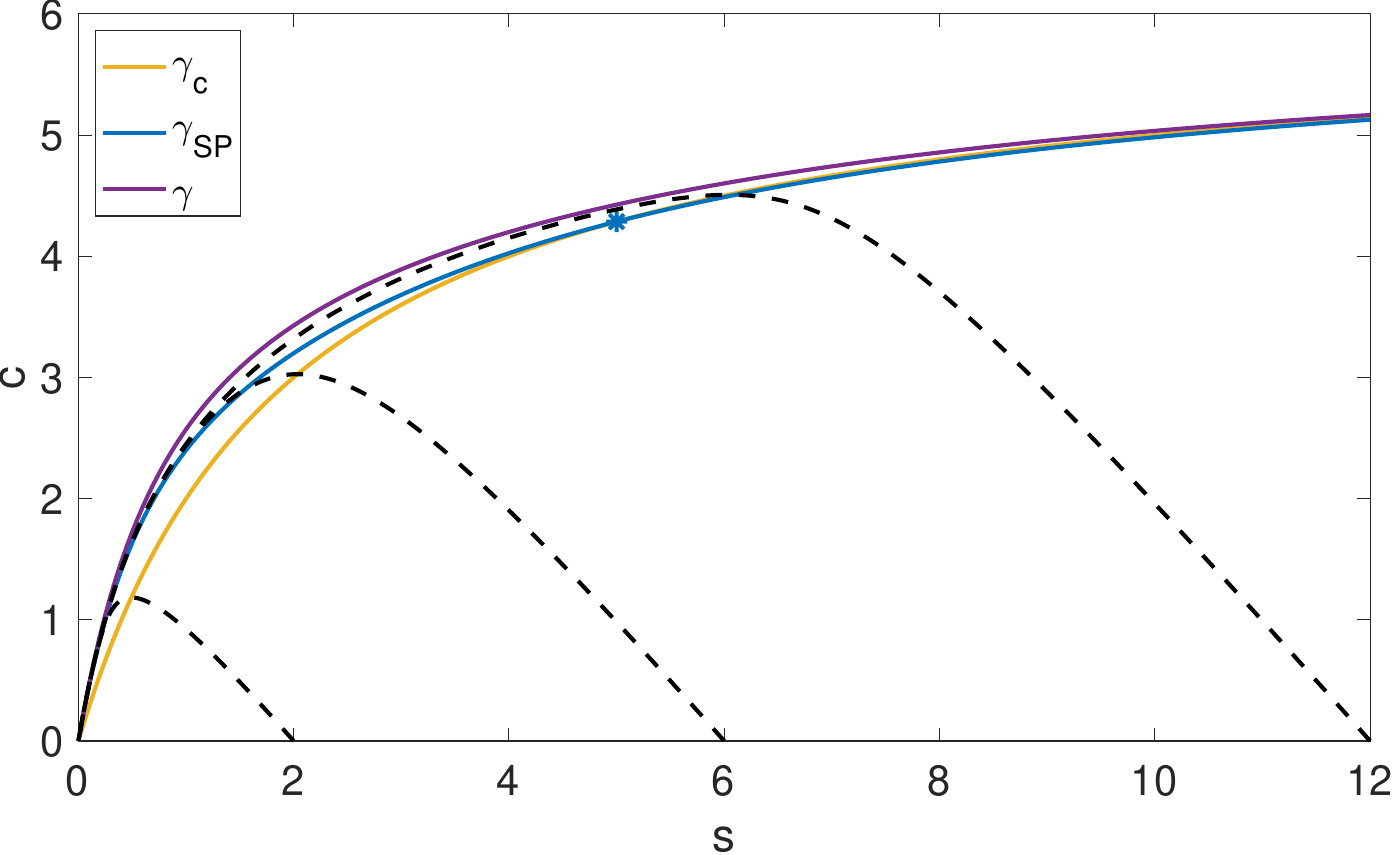}
    \caption{\textbf{The region between the slow product QSS curve $\gamma_{SP}(s)$ 
    and the curve $\gamma(s)$ is a trapping region for solutions in the phase plane 
    when $e_0 < 8K_S$}. The intersection point of $\gamma_{SP}$ and $\gamma_c$, $s^*$, 
    is denoted by the blue star. The numerical solutions to the mass action 
    equations~\eqref{MM} for  several initial conditions are denoted by black dashed 
    lines. $s_0 \in \{2,6,12\}$, $c_0 = 0$, $e_0 = 6$, $k_1 = 1$, $k_{-1} = 1$, 
    $k_2 = 1.$}
    \label{fig:slow_product_solutions}
\end{figure} 

Our main results establish the conditions for positive invariance of the set bordered by the graphs of $\gamma_{SP}(s)$ 
and $\gamma(s)$:
\begin{equation}\label{gamma_SP}
    \Gamma_{SP} := \left\{(s,c) \in \mathbb{R}^2_{\geq 0}\; : \; \;\gamma_{SP}(s) \leq c \leq 
    \gamma(s)\right\},
\end{equation}
and  show that all solutions eventually cross $\gamma_{SP}(s)$ under those conditions.

\begin{theorem}\label{thm:trapping_slowproduct}
For a solution $(s,c)(t)$ to \eqref{mmMA} with initial conditions $(s,c)(0) = (s_0,0)$. 
If $e_0 < 8K_S$, then there exists a $t_{\rm cross\_sp} > 0$ at which the trajectory crosses $\gamma_{SP}(s)$, and  
\begin{equation}
(s,c)(t) \in \Gamma_{SP}, \quad \forall \; t \geq t_{\rm cross\_sp}.
\end{equation}
\end{theorem}

\begin{proof}
To establish positive invariance, we show that the vector field at the boundary curve 
$\gamma_{SP}$ points towards $\Gamma_{SP}$ when $e_0 < 8K_S$:
\begin{equation}
\dot{c}(s,\gamma_{SP}(s)) - \gamma_{SP}'(s)\dot{s}(s,\gamma_{SP}(s)) > 0 \label{eq:ineq1_SP}
\end{equation}
See~\ref{sec:appendix_trapping2} for details. Combining this with \eqref{eq:ineq1} 
ensures solutions entering $\Gamma_{SP}$ remain within.

We also show that solutions beginning outside $\Gamma_{SP}$ eventually enter it. In contradiction, 
assume that a solution starting at $(s_0,0)$ converges to the stationary point $(0,0)$ without
entering $\Gamma_{SP}$. 
Then, the slope of the tangent to $c = \gamma_{SP}(s)$ at $s = 0$ should be greater than
$m$ \eqref{eq:alpha}. However, we show that
\begin{equation}\label{eq:ineq_slowp_slope}
\gamma_{SP}'(0) < m
\end{equation}
always holds, where:
\begin{equation}
\gamma_{SP}'(0) = \frac{k_1e_0}{k_{-1}} - \frac{k_2e_0}{k_{-1}(K_S+e_0)}.
\end{equation}
When $e_0 < K$, $\gamma'_{SP}(0) < \gamma'_c(0)$ holds, and Theorem~\ref{thm:trapping_gamma_0} 
implies $\gamma'_c(0) < m$, demonstrating $\gamma'_{SP}(0) < m$. For $e_0 > K$, see 
\ref{sec:appendix_crossing_slowproduct} for proof of \eqref{eq:ineq_slowp_slope}. Thus, our 
assumption is false, and a $t_{\rm cross\_sp} > 0$ exists when the trajectory crosses 
$\gamma_{SP}(s)$. The proof is complete by combining the two arguments. 
\end{proof}

The primary benefit of establishing the positive invariance of $\Gamma_{SP}$ is that 
it allows for a more precise localization of the slow manifold.

\begin{theorem}
If $e_0 < 8K_S$, then for the differential equation~\eqref{1D}, $\Gamma_{SP}$ is a 
positively invariant, narrowing anti-funnel within which there exists a unique slow 
manifold, $S_0^{\varepsilon}$. 
\end{theorem}

\begin{proof}
Inequality~\eqref{eq:ineq1_SP} ensures $c = \gamma_{SP}(s)$ is a strong upper fence 
for \eqref{1D} when $e_0 < 8K_S$. It is verifiable that $\gamma(s) > \gamma_{SP}(s)$ 
for $s > 0$ and: 
\begin{equation}
\lim\limits_{s \rightarrow \infty} |\gamma(s) - \gamma_{SP}(s)| = 0
\end{equation}
implying $\Gamma_{SP}$~\eqref{gamma} is a narrowing anti-funnel for $e_0 < 8K_S$. 
Further, 
\begin{equation}
    \cfrac{\partial f}{\partial c} (c,s) \geq 0
\end{equation}
in $\Gamma_{SP}$, and the claim follows from 
Theorem~\ref{thm:antifunnel}.
\end{proof}

\begin{figure}
    \centering
     \begin{subfigure}
         \centering
         \includegraphics[scale = 0.38]{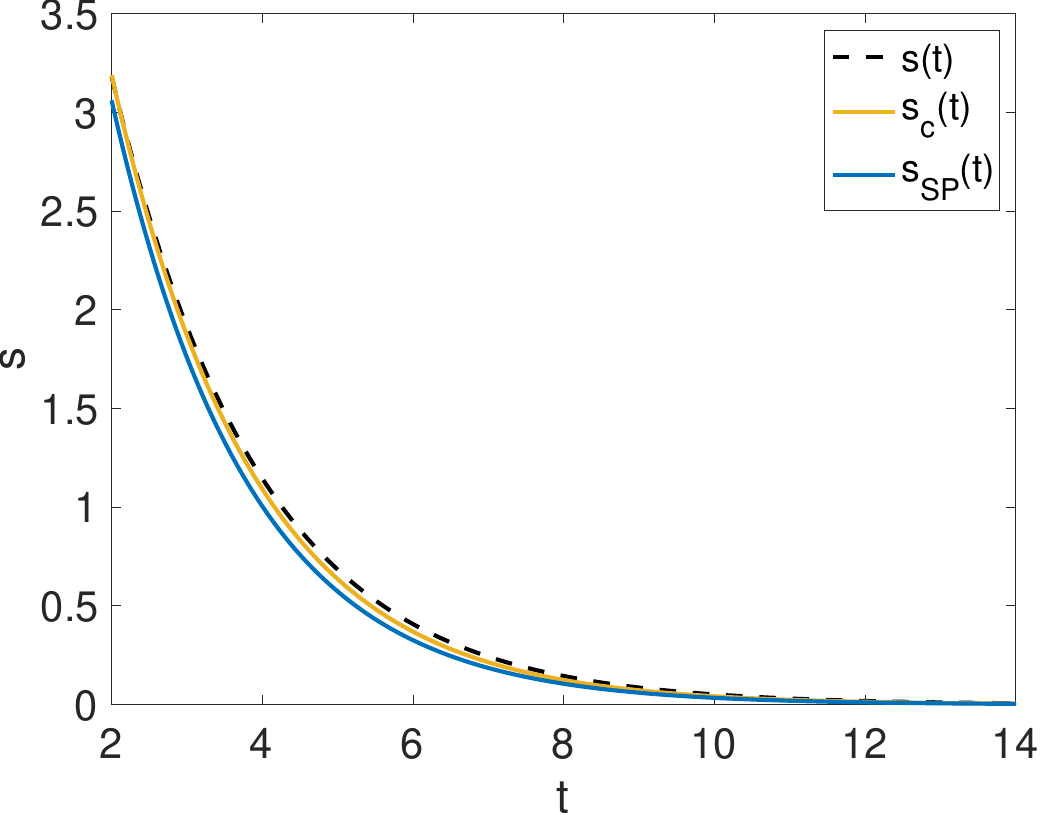}
     \end{subfigure}
     \qquad
     \begin{subfigure}
         \centering
         \includegraphics[scale = 0.38]{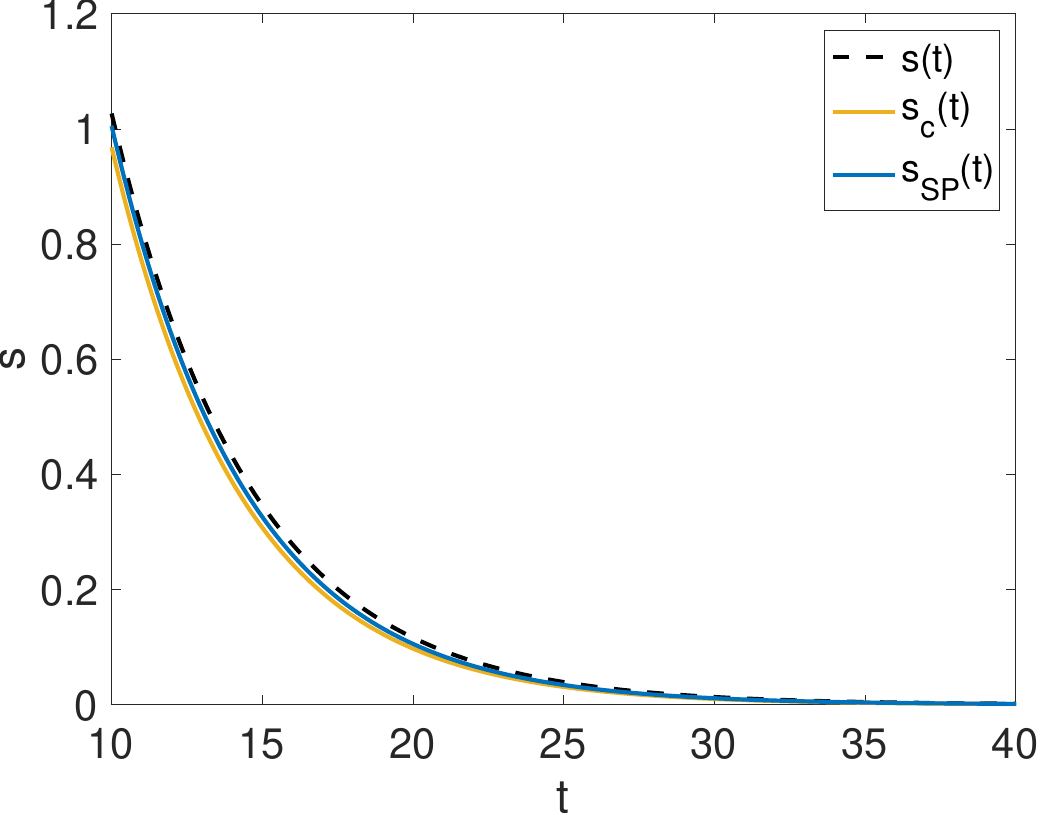}
     \end{subfigure}  
    \caption{\textbf{The condition $e_0 < K$ ensures that the sQSSA provides a better 
    approximation than the slow product QSSA under the Reich-Sel'kov condition.} 
    $e_0 \ll K_M$ is satisfied in both panels. \textbf{Left ($e_0 < K$)}: $s_0 = 9$, 
    $c_0 = 0$, $e_0 = 6$, $k_1 = 1$, $k_{-1} = 100$, $k_2 = 10.$ \textbf{Right ($e_0 > K$)}: 
    $s_0 = 9$, $c_0 = 0$, $e_0 = 6$, $k_1 = 1$, $k_{-1} = 100$, $k_2 = 4$. Substrate 
    depletion over time is shown in both panels. The numerical solution to the mass 
    action equations~\eqref{MM} is denoted by black dashed lines. $s_c(t)$ is the 
    numerical solution to \eqref{sQSSA2} and $s_{SP}(t)$ is the numerical solution 
    to \eqref{slowp}.}
    \label{fig:standard_fails}
\end{figure}

The slow manifold is contained within $\Gamma_{SP}$ when $e_0 < 8K_S$. While this condition 
might seem restrictive, it connects directly to the sQSSA's validity analysis. We show that 
the slow manifold lies closer to the slow product QSS curve than the sQSS curve when 
$e_0 < 8K_S$. The remaining task is to analyze the implications of this location.

\subsection{Revisiting the Reich-Sel'kov Condition}\label{sec:revisiting_RS}
The Reich-Sel'kov qualifier ensures the sQSS reduction's accuracy, as validated in 
Section~\ref{sec:trapping} using the novel anti-funnel \eqref{thm:trapping_gamma_0}. 
Our analysis shows that the sQSSA is the best known approximation to the slow manifold 
under most parametric scenarios. For instance, recall the intersection 
point~\eqref{eq:sp_intersection}. When $e_0 < K$, the intersection is inconsequential. 
Consequently, when the Reich-Sel'kov condition $e_0 \ll K_M$ is combined with $e_0 < K$,
the sQSSA is the predominant approximation.

What happens when $e_0 \ll K_M$, $e_0 \ll K_S$ (ensuring $e_0 < 8K_S$), but $e_0 > K$?
Theorem~\ref{thm:trapping_slowproduct} implies that trajectories lie closer to 
$\gamma_{SP}$ than to $\gamma_c$ for a significant portion of the slow dynamics (see 
{\sc Figure}~\ref{fig:slow_product_solutions}), and therefore 
\begin{equation}
{-\cfrac{k_2e_0s}{K_M+s} \leq -\cfrac{k_2e_0s(K_S+s)}{K_Se_0+(K_S+s)^2} \leq \dot{s} \leq -\cfrac{k_2e_0s(K_M+s)}{K_Me_0+(K_M+s)^2}}
\end{equation}
after the trajectories cross $\gamma_{SP}$ and $s \leq s^*$. Moreover, we know that the 
slope $m$ of the linear approximation is greater than the slope of the tangent to 
$\gamma_{SP}$ at the origin~\eqref{eq:ineq_slowp_slope} and $\gamma_{SP}'(0) > \gamma_c'(0)$ 
whenever $e_0 > K$. Thus, the approximation errors near the stationary point follow:
\begin{equation}
m - \gamma_{SP}'(0) < m - \gamma_c'(0)
\end{equation}
when $e_0 > K$. Consequently, complementing the Reich-Sel'kov 
condition with the stronger condition $e_0 < K$ is necessary to guarantee that the sQSS is 
indeed the predominant quasi-steady-state approximation of the MM reaction 
mechanism, and this may be relevant in the context of the inverse problem where parameters are 
estimated at low substrate concentration~\cite{Stroberg2016}: If $e_0 > K$, using the slow 
product QSS variety as an approximation for 
small $s$ is beneficial. {\sc Figure}~\ref{fig:standard_fails} illustrates the numerical substrate 
depletion curves. The restrictive condition $e_0 < K$ is necessary to ensure the sQSS most 
accurately approximates the MM dynamics for small substrate concentrations.

\section{Discussion}
This work addresses a critical gap in the understanding of the quasi-steady-state approximation 
for the MM reaction mechanism. While the sQSSA is a widely used tool in enzyme kinetics, the 
precise conditions for its validity have remained a topic of ongoing investigation. We conducted 
a thorough phase-plane analysis of the MM mass-action kinetics using the theory of fences and 
anti-funnels, a powerful technique for analyzing the long-time behavior of dynamical systems.

Our analysis led to the identification of new positively invariant sets that contain the slow 
manifold, $S_0^{\varepsilon}$, in the substrate-complex phase plane. These sets provide valuable 
insights into the dynamics of the system and allow for a more precise characterization of the 
sQSSA's accuracy. As a result, we obtained improved bounds on the estimation error of various 
QSS approximations, including the standard, reverse, and slow product formation approximations, 
in the slow regime.

Significantly, we have demonstrated that the commonly accepted qualifier for the validity of 
the sQSSA, the Reich-Sel'kov condition ($e_0 \ll K_M$), does not guarantee that the sQSSA is 
the predominant reduction. Predominance necessitates a more restrictive condition
\begin{equation}
e_0 \ll K,
\end{equation}
where $K$ is the Van Slyke-Cullen constant. This finding challenges the traditional understanding 
of the sQSSA's validity and highlights the importance of considering both accuracy and predominance 
when evaluating QSS approximations.

It is important to note that our analysis primarily focused on the approximation error in the 
slow regime, specifically the error between the Fenichel reduction and the actual flow on the 
slow manifold, $S_0^{\varepsilon}$. However, two primary sources of error contribute to the 
overall accuracy of approximate solutions to singularly perturbed ODEs. The first is the error 
in approximating the flow on the slow manifold itself, as addressed in our analysis. The second 
is the error in approximating the trajectory's approach to the slow manifold, including timescale 
estimates such as $t_{\rm cross}$, which demarcates the intersection of trajectories with 
pertinent QSS varieties. This latter source of error is generally more challenging to analyze. 
Although some progress has been made in this area (see \citet{eilertsen2024rigorous}), further
exploration and improvement is required.

Our findings have broader implications for the application of QSS approximations in various 
fields. By providing a more refined understanding of the sQSSA's validity, our work can guide 
researchers in selecting the most appropriate and accurate reduction for their specific needs. 
This is particularly crucial in areas such as quantitative biology and pharmacology, where accurate 
model reduction is essential for understanding complex biological processes and designing 
effective therapeutic interventions.

Future research could extend our analysis by considering more complex reaction mechanisms 
or incorporating additional factors that might influence the sQSSA's validity. Further 
investigation of the error associated with the trajectory's approach to the slow manifold 
is also warranted. By addressing these open questions, we can continue to refine our 
understanding of QSS approximations and enhance their utility in diverse scientific disciplines.

\section*{Acknowledgments}
KS is supported by a research fellowship from the College of Science at the University of 
Notre Dame.

\section*{Data availibility}
We do not analyse or generate any datasets, because our work proceeds within a theoretical and mathematical approach.

\appendix
\section*{Appendices}
These appendices provide supplementary information to enhance the understanding of the main text. 
The appendices include a synopsis of the theory of fences, funnels, and anti-funnels, as well as detailed 
calculations and proofs for key results presented in the paper.

\section{Synopsis of the theory of fences, funnels and antifunnels}\label{sec:appendix_synopsis}
This section highlights the notation of fences and anti-funnels within the context of the 
Michaelis-Menten system.

In general, an anti-funnel is the region above an upper fence and below a lower fence in a 
two-dimensional phase space. Typically, anti-funnels (or funnels) are visualized with the 
independent variable increasing from left to right. The vector field points outwards at the 
fences of an anti-funnel, while it points inwards at the fences of a funnel. Consequently, 
solutions generally leave anti-funnels. The anti-funnel theory implies that all but one 
solution eventually leave a narrowing anti-funnel \cite[see][Chapter 1]{hubbard1997differential}). 
This property makes fences and anti-funnels useful for identifying a unique, exceptional 
solution that remains within the anti-funnel.

In the context of our work on the Michaelis-Menten system, the substrate concentration, 
$s$, decreases over time, with $s(0) = s_0$ and $s(t) \rightarrow 0$ as $t \rightarrow \infty$.
In the $(s,c)$ phase-plane, solution trajectories start at the substrate axis and move 
from right to left with respect to the $s$-axis. Thus, the vector field always has a negative 
component in the $s$-direction.

Despite this difference in directionality, the concept of anti-funnels remains applicable. 
The vector field can still point inwards at the anti-funnel fences in the $(s,c)$ 
coordinate system as we move from right to left, while satisfying the conditions in 
Definition~\ref{def:fences}. Consequently, solution trajectories can enter the anti-funnels 
defined in the paper.

It is important to remember the nomenclature of upper and lower fences in this context. 
The slow manifold lies above the upper fence and below the lower fence in the phase plane.

\section{\texorpdfstring{$\Gamma_0$}{Gamma} is a positively invariant set for solutions 
in the phase plane}\label{sec:appendix_trapping1}
This section provides the detailed calculations to prove inequality \eqref{eq:ineq1}, 
which establishes the positive invariance of the set $\Gamma_0$. 

We begin by evaluating the dynamics at the curve $c = \gamma(s)$:
\begin{subequations}
    \begin{align}
        \dot{s}(s,\gamma(s)) &= -k_1e_0s + k_1(K_S+s)\gamma(s), \\ 
                &= -k_1e_0s + k_1(K_S+s)\left(\frac{e_0s}{K_S+s} - \frac{k_2e_0s(K_M+s)}{k_1(K_S+s)(K_Me_0 + (K_M+s)^2)}\right), \\
                &= -\frac{k_2e_0s(K_M+s)}{K_Me_0 + (K_M+s)^2},
    \end{align}
\end{subequations}
and
\begin{subequations}
    \begin{align}
        \dot{c}(s,\gamma(s)) &= k_1e_0s - k_1(K_M+s)\gamma(s), \\
                &= k_1e_0s - k_1(K_M+s)\left(\frac{e_0s}{K_S+s} - \frac{k_2e_0s(K_M+s)}{k_1(K_S+s)(K_Me_0 + (K_M+s)^2)}\right), \\
                &= -\frac{k_2e_0s}{K_S+s} + \frac{k_2e_0s(K_M+s)^2}{(K_S+s)(K_Me_0 + (K_M+s)^2)}.
    \end{align}
\end{subequations}

The derivative of $\gamma(s)$ is: 
\begin{align}
    \gamma'(s) &= \frac{K_S e_0}{(K_S+s)^2} - \frac{k_2K_Se_0(K_M+s)}{k_1(K_S+s)^2(K_Me_0 + (K_M+s)^2)}  \\  & \quad - \frac{k_2e_0s}{k_1(K_S+s)(K_Me_0 + (K_M+s)^2)} 
   + \frac{2k_2e_0s(K_M+s)^2}{k_1(K_S+s)(K_Me_0 + (K_M+s)^2)^2}.
\end{align}
The required expression at $c = \gamma(s)$ is:
\begin{subequations}
    \begin{align}
        \dot{c} - \gamma'(s)\dot{s} &= -\frac{k_2e_0s}{K_S+s} + \frac{k_2e_0s(K_M+s)^2}{(K_S+s)(K_Me_0 + (K_M+s)^2)} + \frac{k_2e_0s(K_M+s)}{K_Me_0 + (K_M+s)^2} \gamma'(s), \\
        &= -\frac{k_2e_0s}{(K_S+s)(K_Me_0 + (K_M+s)^2)}\left(K_Me_0 - (K_M+s)(K_S+s)\gamma'(s)\right).
    \end{align}
\end{subequations}
To prove inequality \eqref{eq:ineq1}, we need to show that 
$\delta(s) = K_Me_0 - (K_M+s)(K_S+s)\gamma'(s) > 0$. This also proves that $\gamma(s)$ 
is a lower fence. Substituting $\gamma'(s)$ and simplifying, we obtain:
    \begin{align*}
        \delta(s) &= K_Me_0 - \frac{K_Se_0(K_M+s)}{(K_S+s)} + \frac{k_2K_Se_0(K_M+s)^2}{k_1(K_S+s)(K_Me_0 + (K_M+s)^2)} \\
        & \quad + \frac{k_2e_0s(K_M+s)}{k_1(K_Me_0 + (K_M+s)^2)} - \frac{2k_2e_0s(K_M+s)^3}{k_1(K_Me_0+(K_M+s)^2)^2}, \\
        & = \frac{k_2e_0}{k_1(K_S+s)(K_Me_0 + (K_M+s)^2)^2}\Bigl((K_Me_0 + (K_M+s)^2 s)^2  \\
        & \qquad + K_S(K_M+s)^2(K_Me_0 + (K_M+s)^2) \\ & \qquad \qquad + (K_M+s)(K_S+s)(K_Me_0 + (K_M+s)^2)s  - 2(K_M+s)^3(K_S+s)s\Bigl). 
    \end{align*}
Expanding the first term in the expression, 
\begin{align*}
    \delta(s) &=  \frac{k_2e_0}{k_1(K_S+s)(K_Me_0 + (K_M+s)^2)^2}\Bigl((K_Me_0)^2s + 2(K_Me_0)(K_M+s)^2s + (K_M+s)^4s  \\
        & \qquad + K_S(K_M+s)^2(K_Me_0 + (K_M+s)^2) \\ & \qquad \qquad + (K_M+s)(K_S+s)(K_Me_0 + (K_M+s)^2)s  - 2(K_M+s)^3(K_S+s)s\Bigl), \\
        &= \frac{k_2e_0}{k_1(K_S+s)(K_Me_0 + (K_M+s)^2)^2}\Bigl((K_Me_0)^2s + 2(K_Me_0)(K_M+s)^2s \\
        & \qquad + (k_2/k_1)(K_M+s)^3 s + K_S(K_M+s)^2(K_Me_0 + (K_M+s)^2) \\
        & \qquad \qquad + K_Me_0(K_M+s)(K_S+s)s\Bigl),
\end{align*}
This expression is a sum of positive terms, proving inequality \eqref{eq:ineq1}. 

For inequality \eqref{eq:ineq2}, a straightforward calculation at $c = \gamma_c(s)$ gives: 
\begin{equation*}
    \dot{c} - \gamma_c'(s)\dot{s} = 0 + \left(\frac{K_Me_0}{(K_M+s)^2}\right)\left(\frac{k_2e_0s}{K_M+s}\right) 
\end{equation*}
which is always positive for $s > 0$. Thus, $\Gamma_0$ is a positively invariant region. 

\section{Positive invariance of \texorpdfstring{$\Gamma_{SP}$}{GammaSP} when 
\texorpdfstring{$e_0 < 8K_S$}{e0 < 8KS}}\label{sec:appendix_trapping2}
This section provides the detailed calculations to prove inequality \eqref{eq:ineq1_SP}, 
which establishes the positive invariance of the set $\gamma_{SP}$ when $e_0 < 8K_S$.

We first compute the differential terms at $\gamma_{SP}$:
    \begin{subequations}
    \begin{align}
        \dot{c}(s,\gamma_{SP}(s)) &= k_1e_0s - k_1(K_M+s)\gamma_{SP}(s), \\
                &= k_1e_0s - k_1e_0s\frac{(K_M+s)}{(K_S+s)} + 
                \frac{k_2e_0s(K_M+s)}{K_Se_0 + (K_S+s)^2}, \\
                &= -\frac{k_2e_0s}{(K_S+s)} +
                \frac{k_2e_0(K_M+s)s}{K_Se_0 + (K_S+s)^2}, 
    \end{align}
    \end{subequations}
and 
\begin{subequations}
    \begin{align}
        \dot{s}(s,\gamma_{SP}(s)) &= -k_1e_0s + k_1(K_S+s)\gamma_{SP}(s), \\
                &= -\frac{k_2e_0(K_S+s)s}{K_Se_0 + (K_S+s)^2}.
    \end{align}
\end{subequations}
The expression at $c = \gamma_{SP}(s)$ is:
\begin{subequations}
    \begin{align}
        \dot{c} - \gamma_{SP}'(s)\dot{s}&= -\frac{k_2e_0s}{(K_S+s)} +
                \frac{k_2e_0(K_M+s)s}{K_Se_0 + (K_S+s)^2} + \frac{k_2e_0(K_S+s)s}{K_Se_0 + (K_S+s)^2}\gamma_{SP}'(s), \\
                &= \frac{k_2e_0s}{(K_S+s)(K_Se_0 + (K_S+s)^2)} \bigl(-(K_Se_0 + (K_S+s)^2) + \\
                & \qquad (K_M+s)(K_S+s) + (K_S+s)^2\gamma_{SP}'(s) \bigr)
    \end{align}
    \end{subequations}
Proving inequality \eqref{eq:ineq1_SP} simplifies to showing:
\begin{equation}
    -(K_Se_0 + (K_S+s)^2) +
                 (K_M+s)(K_S+s) + (K_S+s)^2\gamma_{SP}'(s) > 0.
\end{equation}
This also proves that $\gamma_{SP}(s)$ is an upper fence. Substituting $\gamma_{SP}'(s)$:
\begin{subequations}
    \begin{align}
        \gamma_{SP}'(s) = \frac{K_Se_0}{(K_S+s)^2} - \frac{k_2e_0}{k_1(K_Se_0 + (K_S+s)^2)} + \frac{2k_2e_0(K_S+s)s}{k_1(K_Se_0 + (K_S+s)^2)^2},
    \end{align}
\end{subequations}
and simplifying, we obtain:
\begin{equation}
    \frac{k_2}{k_1}(K_S+s) - \frac{k_2e_0(K_S+s)^2}{k_1(K_Se_0 + (K_S+s)^2)} + \frac{2k_2e_0(K_S+s)^3s}{k_1(K_Se_0 + (K_S+s)^2)^2}
\end{equation}
which is positive for $e_0 < 8K_S$. To show this, we further simplify the inequality to:
\begin{subequations}
\begin{align}
    (K_S+s)^4 + e_0(K_S+s)^3 - K_Se_0^2s &> 0, \\
    1 + \frac{e_0}{(K_S+s)} - \frac{K_Se_0^2s}{(K_S+s)^4} &> 0.
    \end{align}
\end{subequations}
Substituting $e_0 = \eta K_S$ and introducing $x = K_S/(K_S+s)$, we get: 
\begin{equation}
   1 + \frac{e_0}{(K_S+s)} - \frac{K_Se_0^2s}{(K_S+s)^4} = 1 + \eta x - \eta^2 x^3(1-x) = \eta^2x^4 - \eta^2x^3 + \eta x + 1.
\end{equation}
We want to show that $\eta^2x^4 - \eta^2x^3 + \eta x + 1 > 0$ for $\eta < 8$. Dividing 
by $\eta^2$ and factoring the quartic function, we get:
\begin{equation}
    x^4 - x^3 + \frac{1}{\eta} x + \frac{1}{\eta^2} = \left(x^2 + \bar{\eta}_+x - \frac{1}{\eta}\right)\left(x^2 + \bar{\eta}_-x - \frac{1}{\eta}\right)
\end{equation}
where 
$$\bar{\eta}_{\pm} = -\frac{1}{2} \pm \frac{1}{2}\sqrt{\frac{\eta - 8}{\eta}}.$$ 
This quartic function has no real roots when $\eta < 8$ and remains positive. It has 2 
real repeated roots for $\eta = 8$ and 4 real roots when $\eta > 8$. This completes the 
proof for the first part of Theorem~\ref{thm:trapping_slowproduct}.

\section{Solutions approach the origin at a slope greater than \texorpdfstring{$\gamma_{SP}'(0)$}{slowp product}. }\label{sec:appendix_crossing_slowproduct}
This section provides the details to prove the inequality:
\begin{equation}
    m - \gamma'_{SP}(0) > 0
\end{equation}
when $k_2 < k_1e_0$, completing the proof for Theorem \ref{thm:trapping_slowproduct}. 

Recall:
\begin{equation}
    m = \frac{1}{2k_{-1}}\left(-k_{-1} - k_2 + k_1e_0 + \sqrt{(k_{-1} + k_2 + k_1e_0)^2 - 4k_1k_2e_0}\right)
\end{equation}
    and 
    \begin{equation}
    \gamma_{SP}'(0) = \frac{k_1e_0}{k_{-1}} - \frac{k_2k_1e_0}{k_{-1}(k_1e_0+k_{-1})}.
\end{equation}
Therefore, $m - \gamma'_{SP}(0)$ is:
\begin{subequations}
    \begin{align}
        m - \gamma_{SP}'(0) &= \frac{-(k_{-1}+k_2+k_1e_0)}{2k_{-1}}  + \frac{k_2k_1e_0}{k_{-1}(k_1e_0+k_{-1})} \\
        & \qquad \qquad + \frac{1}{2k_{-1}}\sqrt{(k_{-1} + k_2 + k_1e_0)^2 - 4k_1k_2e_0}, \\
        &= \frac{-(k_{-1}+k_1e_0)}{2k_{-1}} + \frac{k_2(k_1e_0 - k_{-1})}{2k_{-1}(k_1e_0+k_{-1})} + \frac{1}{2k_{-1}}\sqrt{(k_{-1} + k_2 + k_1e_0)^2 - 4k_1k_2e_0}, \\
        &= \frac{-(k_{-1}+k_1e_0)}{2k_{-1}} + \frac{k_2(k_1e_0 - k_{-1})}{2k_{-1}(k_1e_0+k_{-1})} \\
        & \qquad \qquad  + \frac{1}{2k_{-1}}\sqrt{(k_{-1} + k_1e_0)^2 + k_2^2 + 2k_2k_{-1} - 2k_1k_2e_0}.
    \end{align}
\end{subequations}
Simplifying this expression, we want to show:
\begin{equation}
    \frac{-(k_{-1}+k_1e_0)}{2k_{-1}} + \frac{k_2(k_1e_0 - k_{-1})}{2k_{-1}(k_1e_0+k_{-1})} 
+ \frac{1}{2k_{-1}}\sqrt{(k_{-1} + k_1e_0)^2 + k_2^2 - 2k_2(k_1e_0 - k_{-1})} > 0
\end{equation}
or 
\begin{equation}
\frac{1}{2k_{-1}}\sqrt{(k_{-1} + k_1e_0)^2 + k_2^2 - 2k_2(k_1e_0 - k_{-1})} >  \frac{(k_{-1}+k_1e_0)}{2k_{-1}}  - \frac{k_2(k_1e_0 - k_{-1})}{2k_{-1}(k_1e_0+k_{-1})}  .
\end{equation}
The above expression is equivalent to:
\begin{multline}
    (k_1e_0+k_{-1})\sqrt{(k_{-1} + k_1e_0)^2 + k_2^2 - 2k_2(k_1e_0 - k_{-1})} > \\
    (k_1e_0+k_{-1})^2 - k_2(k_1e_0 - k_{-1}).
\end{multline}
When $k_2 < k_1e_0$, both sides are positive. Thus, we can prove that the above inequality 
is true by showing that the corresponding squared terms satisfy the same inequality. In 
simpler words, if $a, b > 0$ and $a^2 > b^2$, then $a > b$. Squaring both sides and 
simplifying, we can verify that: 
\begin{multline}
    (k_1e_0+k_{-1})^2((k_{-1} + k_1e_0)^2 + k_2^2 - 2k_2(k_1e_0 - k_{-1})) > \\
    (k_1e_0+k_{-1})^4 + k_2^2(k_1e_0 - k_{-1})^2 - 2k_2(k_1e_0 - k_{-1})(k_1e_0+k_{-1})^2.
\end{multline}
Hence, the inequality $m - \gamma_{SP}'(0) > 0$ holds when $k_2 < k_1e_0$.\

\section{Singular perturbation analysis of large \texorpdfstring{$k_{-1}$}{kminus1}}\label{sec:largekminus1}
In perturbation form, the Michaelis-Menten mass action equations with large $k_{-1}$ is
\begin{subequations}
\begin{align}
\dot{s} &= -k_1(e_0-c)s+\varepsilon^{-1}\widehat{k}_{-1}c,\\
\dot{c} &= \;\;\;k_1(e_0-c)s-k_2c-\varepsilon^{-1}\widehat{k}_{-1}c,
\end{align}
\end{subequations}
which admits the corresponding layer problem:
\begin{subequations}
\begin{align}
\dot{s} &= \varepsilon^{-1}\widehat{k}_{-1}c,\\
\dot{c} &= -\varepsilon^{-1}\widehat{k}_{-1}c.
\end{align}
\end{subequations}
The critical manifold is therefore $s$-axis and is normally hyperbolic and attracting. The projection matrix, $\pi^s$, is given by
\begin{equation}
\pi^s = \begin{pmatrix}1 & 1 \\ 0 & 0\end{pmatrix}.
\end{equation}
The Fenichel reduction is therefore
\begin{equation}
\begin{pmatrix}\dot{s}\\\dot{c}\end{pmatrix} = \begin{pmatrix}1 & 1 \\ 0 & 0\end{pmatrix}\cdot\begin{pmatrix}-k_1(e_0-c)s+\varepsilon^{-1}\widehat{k}_{-1}c\\k_1(e_0-c)s-k_2c-\varepsilon^{-1}\widehat{k}_{-1}c\end{pmatrix}\Bigg|_{c=0} = \begin{pmatrix}0\\0\end{pmatrix}
\end{equation}
which is trivial. Thus, one must appeal to higher-order terms in order to recover a non-trivial QSSA. 

\section{QSS approximations and initial conditions}
\label{sec:appendix_initialdata}
This section briefly discusses how to determine appropriate initial conditions for quasi-steady-state approximations. For planar 
systems with a normally hyperbolic and attracting critical manifold, $S_0$, the two-dimensional stable manifold of $S_0$, denoted 
as $W^s(S_0)$, consists of one-dimensional manifolds called {\it fast} fibers. Each point $p_{\mathcal{B}}\in S_0$ is a base 
point,  representing a steady-state (equilibrium) solution to the layer problem. Furthermore, for every base point in $S_0$, 
there exists a corresponding fiber $\mathcal{F}_0(p_{\mathcal{B}})$, which is the stable manifold of that base point:
\begin{equation*}
\mathcal{F}0(p{\mathcal{B}}) =: W^s(p_{\mathcal{B}}), \quad W^s(S_0) = \bigcup_{p_{\mathcal{B}}\in S_0}\mathcal{F}0(p{\mathcal{B}}).
\end{equation*}
These base points serve as approximate initial conditions for quasi-steady-state reductions. Therefore, the task is to identify 
the unique base point that a given trajectory approaches during the transient (fast) phase of the reaction and then use this 
base point as the initial condition for the QSS reduction.

As an illustrative example, consider the Michaelis-Menten mass action system with small $k_2$. The perturbation 
form of the mass action system is:
\begin{subequations}\label{massAC}
\begin{align}
\dot{s} &= -k_1(e_0-c)s+k_{-1}c\
\dot{c} &= ;;;k_1(e_0-c)s-k_{-1}c-\varepsilon \widehat{k}_2c,
\end{align}
\end{subequations}
and the corresponding layer problem is 
\begin{subequations}
\begin{align}
\dot{s} &= -k_1(e_0-c)s+k_{-1}c,\
\dot{c} &= ;;;k_1(e_0-c)s-k_{-1}c.
\end{align}
\end{subequations}
The critical manifold is: 
\begin{equation}
    S_0 := \left\{(s,c)\in \mathbb{R}^2_{\geq 0}: c = \frac{e_0s}{K_S + s}\right\}
\end{equation}
and each base point, $p_{\mathcal{B}}$, can be represented as:
\begin{equation}
p_{\mathcal{B}} =  (s,e_0s/(K_S+s)), \quad \forall s\in[0,s_0].
\end{equation}

The typical initial condition for the Michaelis--Menten reaction mechanism is $(s,c)(0)=(s_0,0)$. As mentioned
previously, the Fenichel reduction for small $k_2$ is:
\begin{equation}\label{slowp2}
\begin{pmatrix}
\dot{s}\\\dot{c}
\end{pmatrix} = -\cfrac{k_2e_0 s(K_S+s)}{K_Se_0 + (K_S+s)^2}\begin{pmatrix} 1\\ \cfrac{K_Se_0 - K(K_S+s)}{(K_S+s)^2}\end{pmatrix}.
\end{equation}
A key question is whether we can equip \eqref{slowp2} with the same initial condition, $(s,c)(0)=(s_0,0)$, and still expect the solution to the mass action system to converge to the solutions of \eqref{slowp2} as $k_2\to 0$.

The answer is no. The reduction~\eqref{slowp2}  approximates the dynamics on the slow manifold; it does not 
describe the dynamics of the approach to the slow manifold. While solutions to the mass action equations~\eqref{massAC} capture 
both the fast and slow stages of the reaction, the QSS reduction~\eqref{slowp2} only describes the slow stage. Consequently,
\eqref{slowp2} requires a ``modified'' initial condition, $\tilde{x}_0=(\tilde{s},\tilde{c})$. This modified initial 
condition is the base point of the fast fiber containing $(s_0,0)$. In this example, the fiber containing $(s,c)(0)=(s_0,0)$ 
is:
\begin{equation}
\mathcal{F}_0(p_{\mathcal{B}}) =\{(s,c)\in \mathbb{R}^2_{\geq 0}: s+c=s_0\}
\end{equation}
and the modified initial condition is the intersection of $\mathcal{F}_0(p_{\mathcal{B}})$ with $S_0$:
\begin{equation}
\tilde{x}_0 = (s,e_0s/(K_S+s)), \quad \text{such that $s$ satisfies:}\;\; s+ \cfrac{e_0s}{K_S+s}=s_0,
\end{equation}
which can be solved algebraically using the quadratic formula:
\begin{equation}
    \tilde{s} = \cfrac{1}{2}\left(s_0-e_0-K_S+\sqrt{(s_0-e_0-K_S)^2+s_0K_S}\right), \quad \tilde{c}= e_0\tilde{s}/(K_S+\tilde{s}).
\end{equation}

In summary, if we equip the reduction~\eqref{slowp2} with the modified initial 
condition~$(s,c)(0)=(\tilde{s},\tilde{c})$, then the solution to this system will converge to the solution of the mass 
action system~\eqref{massAC} (also equipped with the initial condition~$(s,c)(0)=(\tilde{s},\tilde{c})$) as $k_2\to 0$. 

Next, consider the case for small $e_0$, which is more complex than the case for small $k_2$. As previously 
mentioned, the critical manifold is the $s$-axis, but the QSS variety is actually the $c$-nullcline (a distinction  not fully
understood until recently; see \citet{Noethen2011} for details). The base points are $p_{\mathcal{B}}=(s,0)$, and the fibers
are are more intricate (see \citet[Section 3.1]{Noethen2011}, for a discussion of the first integral of the layer problem 
associated with small $e_0$). However, since $(s_0,0)\in S_0$, we can equip the MM rate law,
\begin{equation}
\dot{s} = -\cfrac{k_2e_0s}{K_M+s}
\end{equation}
with the initial condition $s(0)=s_0$.

Initial conditions for the complex, $c$, can also be computed by analyzing the transient dynamics, but we 
will not delve into those details here. Interested readers can consult~\citet{Segel1989} as well as~\citet{Lax2020}.










\end{document}